\documentclass{article} 
\usepackage{iclr2027_conference,times}

\usepackage{amsmath,amsfonts,bm}

\def\eqref#1{equation~\ref{#1}}

\def\1{\bm{1}}

\DeclareMathAlphabet{\mathsfit}{\encodingdefault}{\sfdefault}{m}{sl}
\SetMathAlphabet{\mathsfit}{bold}{\encodingdefault}{\sfdefault}{bx}{n}

\usepackage{hyperref}
\usepackage[capitalize]{cleveref}
\usepackage{url}
\usepackage{xspace}
\usepackage{graphicx}
\usepackage{xcolor,soul}
\usepackage{booktabs}
\usepackage{amsthm}
\usepackage{dsfont}
\usepackage{makecell}
\usepackage{graphicx}
\usepackage{colortbl}
\usepackage{algorithmicx}
\usepackage{algpseudocode}
\usepackage{caption}
\usepackage{array}
\usepackage{multirow}
\usepackage{longtable}
\usepackage{wrapfig}

\title{\semfield: A Simple, Linear, Continuous, \textit{yet} Robust Semantic Watermark}

\author{
  \parbox{0.9\linewidth}{\centering 
    Varun Gumma$^{1}$\thanks{Corresponding Author: \texttt{varun.gumma@protonmail.com}} \quad 
    Navonil Majumdar$^{1}$ \quad 
    Soujanya Poria$^{1}$ \\[0.5em]
    $^{1}${\rm DeCLaRe Lab, Nanyang Technological University} 
  }
}
\newcommand{\semfield}{\textsc{SemField}\xspace}
\newcommand{\semfieldpl}{\textsc{SemField-PL}\xspace}
\newcommand{\alimark}{\textsc{AliMark}\xspace}
\newcommand{\samark}{\textsc{SAMark}\xspace}
\newcommand{\KGW}{\textsc{KGW}\xspace}
\newcommand{\EXP}{\textsc{EXP}\xspace}
\newcommand{\SIR}{\textsc{SIR}\xspace}  
\newcommand{\PRO}{\textsc{PRO}\xspace}
\newcommand{\synthid}{\textsc{SynthID}\xspace}
\newcommand{\pmark}{\textsc{PMark}\xspace}
\newcommand{\simmark}{\textsc{SimMark}\xspace}
\newcommand{\semstamp}{\textsc{SemStamp}\xspace}
\newcommand{\ksemstamp}{$k$-\textsc{SemStamp}\xspace}
\newcommand{\swordstamp}{\textsc{SwordStamp}\xspace}
\newcommand{\kswordstamp}{$k$-\textsc{SwordStamp}\xspace}

\newcommand{\opt}{\textsc{OPT-1.3B}\xspace}
\newcommand{\llama}{\textsc{Llama-3.2-1B}\xspace}
\newcommand{\llamajudge}{\textsc{Llama-3.1-8B}\xspace}
\newcommand{\qwen}{\textsc{Qwen3-1.7B}\xspace}
\newcommand{\qwenrewriter}{\textsc{Qwen3.8-27B}\xspace}
\newcommand{\mpnet}{\textsc{all-mpnet-base-v2}\xspace}

\newcommand{\clean}{\textsc{Clean}\xspace}
\newcommand{\pegasus}{\textsc{Pegasus}\xspace}
\newcommand{\parrot}{\textsc{Parrot}\xspace}
\newcommand{\dipper}{\textsc{Dipper}\xspace}
\newcommand{\llmrewrite}{\textsc{LLM-Rewrite}\xspace}
\newcommand{\realnews}{\textsc{C4-RealNewsLike}\xspace}

\newcommand{\tpr}{$\operatorname{TPR}$\xspace}
\newcommand{\tpro}{$\operatorname{TPR@1\%}$\xspace}
\newcommand{\tprf}{$\operatorname{TPR@5\%}$\xspace}
\newcommand{\fpr}{$\operatorname{FPR}$\xspace}
\newcommand{\auroc}{$\operatorname{AUROC}$\xspace}
\newcommand{\pplratio}{$\operatorname{PPLRatio}$\xspace}

\newcommand{\best}[1]{\textbf{#1}}
\newcommand{\second}[1]{\underline{#1}}

\definecolor{mygreen}{HTML}{d9ead3}
\definecolor{myorange}{HTML}{fce5cd}
\definecolor{myred}{HTML}{f4cccc}
\definecolor{mymagenta}{HTML}{ead1dc}
\definecolor{myblue}{HTML}{cfe2f3}
\definecolor{mygray}{HTML}{efefef}
\definecolor{pert_green}{HTML}{02a61b}
\definecolor{pert_red}{HTML}{cc0000}

\newtheorem{theorem}{Theorem}

\newtheorem{proposition}{Proposition}

\newcommand{\ruledalgorithm}[3]{%
\begin{minipage}[t]{#1}
\hrule
\vspace{4pt}
\textbf{#2}
\vspace{3pt}
\hrule
\vspace{4pt}
#3
\vspace{4pt}
\hrule
\end{minipage}
}

\definecolor{nmcolor}{RGB}{7, 136, 66}

\iclrfinalcopy 
\begin{document}
\maketitle

\begin{abstract}
The rapid proliferation of Large Language Models (LLMs) necessitates reliable watermarking techniques to identify AI-generated text and ensure appropriate attribution. While token-based watermarks are vulnerable to paraphrasing, a central challenge for semantic watermarking is to turn sentence meanings into a stable, well-calibrated document-level signal. To this end, we introduce \semfield, a simple, training-free semantic watermark that embeds a continuous, linear signal directly into the sentence embedding space. Using a shared secret key to define a specific Gaussian direction, \semfield iteratively evaluates candidate sentences and selects those that maximize the alignment of the cumulative document aggregate with this targeted direction. We also propose \semfieldpl, a polarity-locked variant that first determines the optimal orientation from the initial sentence and continuously reinforces it throughout the generation process. The document-level aggregated formulation guarantees exact invariance to sentence reordering and provides theoretical bounds against structural tampering, such as sentence insertion and deletion. Lastly, with extensive evaluations across three models, we demonstrate that both variants outperform 12 recent baselines. Across clean detection and four paraphrasing attacks, both variants achieve a mean True Positive Rate (\tpr) of 88.2\% to 90.4\% at a 1\% False Positive Rate (\fpr), all while maintaining comparable perplexity and naturalness as human-generated content. We open-source our implementation at \url{https://github.com/declare-lab/SemField}
\end{abstract}
\section{Introduction}
\label{sec:intro}

The widespread use of LLMs has coincided with a rapid increase in AI-generated text online. Previous studies \citep{hanley2024machinemade,dolezal2026impact} have reported growth in machine-generated news articles and in AI-generated or assisted material on newly published websites. Similarly, there has also been a boost in AI-written research ideas \citep{ICLR2025_ea94957d}, papers \citep{Liang2025} and large-scale codebases \citep{doi:10.1126/science.adz9311} which users can claim as their own \citep{10.1145/3771937}. In such cases, providers must ascertain if a text was generated using their models for appropriate ownership attribution. 

Watermarking has emerged as one such algorithmic safeguard that embeds imperceptible statistical signals directly into model generations to allow post-hoc identification. To this end, major companies such as Google and Anthropic have already declared the inclusion of a watermark in their outputs for downstream identification \citep{Dathathri2024,Stutz2026,anthropic2026claudewatermark}, while regulators have proposed norms for machine-readable marking and detection of AI-generated content \citep{ec2026codeofpractice}. A reliable watermarking algorithm not only provides a strong detection on its generated documents, but also avoids falsely claiming independent human-generated or other providers' text under their own watermark. In recent times, AI-generated content is also rarely used as is and is mostly subjected to local editing or full paraphrasing before further usage. Hence, a strong watermark signal should be sufficiently robust to such tampering while still providing a high true-positive yet low false-positive identification rate. Lastly, the inclusion of the watermark should not degrade the quality of the associated text and should be as close as possible to the natural text generated by the model or human. 

Early language-model watermarks biased token sampling according to a secret key, making certain tokens more likely and testing whether they occur unusually often \citep{pmlr-v202-kirchenbauer23a,Dathathri2024,aaronson2023watermarking}. However, a rewrite can replace those tokens while retaining the meaning of the passage. To counter such paraphrasing attacks, Semantic watermarks move the signal to sentence representations and select candidates using semantic regions, relations between neighboring sentences, or structured codes \citep{hou-etal-2024-semstamp,hou-etal-2024-k,dabiriaghdam-wang-2025-simmark,li2026alimark}. These methods show how watermark evidence can be tied to meaning rather than exact wording. 

In this work, we further the idea of semantic watermarks by exploring whether a \textit{continuous signal} over the embedding space can be successfully accumulated at the document level. Following this motivation, we propose \semfield\footnote{We use ``field'' to mean a continuous scalar function over the sentence-embedding space}, a simple, linear, and training-free semantic watermark. \semfield represents the watermark signal as a continuous scalar function over normalized sentence embeddings. The mechanism relies on a secret key that determines a specific Gaussian direction within the embedding space. During generation, \semfield iteratively samples candidate next sentences, and projects them onto this keyed direction to calculate their alignment. It then selects the sentence that best increases the score of the resulting \textit{prospective} document aggregate.

Since the Gaussian vector is symmetric, the candidates can be selected to align along either orientation. Hence, we extend the logic to propose a polarity-locked variant, \semfieldpl. This extension identifies the orientation to align based on the first sentence of the document, subsequently reinforcing the signal against that specific orientation throughout the generation process. At the detection stage for both variants, the sentence embeddings are summed together, and the normalized aggregate is projected onto the same keyed direction. This yields a robust document-level statistic, which is then evaluated against a standard normal distribution to reliably detect the presence of the watermark. Such a formulation based on document-level aggregation also provides some theoretical bounds against structural modifications such as the insertion or deletion of sentences. Further, the summation-based aggregation provides an exact invariance to sentence-level reordering.

Finally, we evaluate \semfield and \semfieldpl against 12 recent baselines, 4 paraphrasing, and 4 structural attacks. To confirm the efficacy of our approach, we also repeat it across 3 different models, and also note the perplexity and MAUVE scores of the generations to ascertain the quality of generations in terms of fluency and naturalness\footnote{as close as possible to human generations}. Across all these experiments, both variants achieve a mean \tpro of $88.2$--$90.4\%$ compared with $77.6$--$81.7\%$ for the strongest baseline on each model. \semfieldpl also consistently achieves the best mean \tprf, reaching $96.4$--$96.6\%$ compared with $91.9$--$94.2\%$ for the strongest baseline on each model. The MAUVE scores of the two variants also rank in the top two on all three models, while their perplexity remains within the range of the compared methods, indicating no abnormal generations. The structural tests further indicate that both variants preserve their clean detection scores under sentence reordering and also remain detectable after substantial sentence insertion and deletion. In short, our contributions are as follows:
\begin{enumerate}
\item We introduce \semfield, a simple, training-free watermark that uses a continuous keyed field and a normalized \textit{prospective} document aggregate to select sentences that best improve the document-level signal. We also present \semfieldpl, an extension that first selects a particular field orientation and then consistently reinforces against it.

\item We derive the conditional Gaussian null distribution under the ideal key model, prove exact invariance to reordering intact sentences, and bounded score changes under sentence insertion and deletion.

\item We provide a thorough comparison of \semfield and \semfieldpl against a list of recent and relevant baselines and demonstrate their strong efficacy and robustness along paraphrasing and structural attacks, while maintaining a comparable textual quality as the baselines. We follow it up with two practical deployment scenarios and once again corroborate the high performance and reliability of our methods. Lastly, we close with the ablations of the major components of our techniques, such as polarity locking, candidate count, and sampling temperature. 
\end{enumerate}
\section{Related Work}
\label{sec:related_work}

\paragraph{Token-based watermarks.}
Most early language-model watermarks operate directly on token probabilities. A secret key partitions or scores the vocabulary at each generation step, and decoding is modified so that selected tokens occur more frequently than expected under unwatermarked generation. The green-list construction of \KGW is a representative example \citep{pmlr-v202-kirchenbauer23a}. Other token-based methods modify the sampling rule without applying a fixed logit bias. \EXP uses keyed pseudorandom values to perform exponential sampling from the language-model distribution \citep{aaronson2023watermarking}. \citet{kuditipudi2024robust} develop distortion-free watermarking through inverse-transform and exponential-minimum sampling, together with alignment-based detection for robustness to token edits. \synthid uses keyed tournament sampling and aggregates the resulting token-level scores during detection \citep{Dathathri2024}.

\paragraph{Semantic watermarks.}
Semantic watermarking moves the signal from individual tokens to choices among semantically related textual realizations. A common approach embeds candidate sentences and uses a secret partition, cluster assignment, or geometric constraint to determine which candidates carry valid watermark evidence. \semstamp \citep{hou-etal-2024-semstamp} and \ksemstamp \citep{hou-etal-2024-k} instantiate this idea using locality-sensitive hashing and learned semantic clusters, respectively. Other methods encode the signal through relations between successive representations, as in \simmark \citep{dabiriaghdam-wang-2025-simmark}, or use context-dependent semantic constraints, as in \pmark \citep{huo2026pmark}. More recent work has explored richer semantic codes and alternative units of representation. \alimark converts projections of sentence embeddings into short binary codes and aggregates evidence across a document \citep{li2026alimark}. \samark constructs self-anchored semantic regions that do not depend on the immediately preceding sentence \citep{huo2026samarkselfanchoredtextwatermarking}. \swordstamp and \kswordstamp move from complete sentences to smaller semantic units and introduce order-robust aggregation \citep{diaa2026semanticwatermarkingorderrobustdetection}. Post-hoc approaches attempt to modify an already generated document using semantically selected lexical items, allowing watermarking without access to model logits \citep{chang-etal-2024-postmark}.

\paragraph{Learning-based watermarks.}
A third family learns some part of the watermarking mechanism from data. The learned component may be a transformation that converts context into a token-level bias, a detector trained jointly with the generator, or a watermark encoded directly into the model parameters. For example, \SIR learns a semantic transformation that maps similar contexts to similar vocabulary biases, thereby connecting semantic invariance with token-level generation \citep{liu2024a}. Other approaches jointly train a language model and detector, often using reinforcement learning, so that generated text becomes detectable while retaining task utility \citep{xu2025learning}. Learned encoder--decoder frameworks have also been used to embed and recover multi-bit messages from generated text \citep{3698900.3699002}. Recent work further considers watermarks for open-source models, where decoding-time control cannot be assumed after model release. \PRO jointly optimizes a watermark policy and the language model so that the learned generation pattern agrees with its detector, while regularization targets robustness to subsequent model modifications such as fine-tuning and model merging \citep{xue2025proenablingpreciserobust}.
\section{Methodology}
\label{sec:methodology}

In this section, we first introduce \semfield, and then extend it to its two-sided \semfieldpl along with all the necessary assumptions, bounds, and proofs. The exact algorithms are provided in Appendix \ref{app:semfield-algo}.

\subsection{Setup}
\label{sec:setup}

Let a generated document $D$ be a sequence of sentences $s_i$:
\[
D=(s_1,\ldots,s_T).
\]

A pre-trained encoder $\mathcal{E}_{\theta}(\cdot)$ is used to obtain normalized sentence embeddings
\begin{equation}
x_t
=
\frac{\mathcal{E}_{\theta}(s_t)}{\lVert \mathcal{E}_{\theta}(s_t)\rVert_2}
\in S^{d-1} \subset \mathbb{R}^{d}.
\end{equation}

Similar to prior zero-bit watermarks, the detector receives a document $D$ and a secret key $K$ and tests the hypothesis
\begin{equation}
H_0:\ D\text{ was generated independently of the watermarking procedure keyed by }K.
\end{equation}
The detector computes a $p$-value and declares a watermark present\footnote{That is, it rejects $H_0$.} when the $p$-value falls below a threshold chosen for the desired \fpr.

\subsection{Keyed Semantic Field}
\label{sec:keyed-semantic-field}

The secret key determines a $d$-dimensional vector through a Gaussian pseudo-random generator:
\begin{equation}
\omega_{K} = \operatorname{GaussianPRG}(K)\sim\mathcal{N}(0,I_d),
\end{equation}
where $I_d$ is the $d$-dimensional identity matrix. The vector $\omega_K$ carries the semantic watermark, and the generation favors sentences that strengthen the document's alignment with that direction.

For a normalized semantic representation $u\in S^{d-1}$, we define the \textit{keyed field value} as the projection,
\begin{equation}
f_{K}(u)
=
\omega_{K}^{\mathsf T}u.
\end{equation}

For fixed normalized vector $u$, the score is standard normal (\Cref{thm:gaussian-field-marginal}):
\begin{equation}
f_{K}(u)\sim\mathcal{N}(0,1).
\end{equation}

Further, the field values of two fixed normalized vectors $u$ and $v$ has covariance (\Cref{thm:covariance-of-keyed-semantic-field-scores})
\begin{equation}
\operatorname{Cov}_{K}
\bigl(f_{K}(u),f_{K}(v)\bigr)
=
u^{\mathsf T}v.
\end{equation}
Since their correlation equals their inner product, close representations tend to receive similar field values. This definition also maintains the \textit{continuity} in the field where minor changes in the embedding also produce a corresponding bounded change in the score (\Cref{prop:field-continuity}).

\subsection{Document-Level Semantic Statistic}
\label{sec:doc-level-sem-stats}

The previous section defined the score of a single sentence $s$ with normalized embedding $x$. Given a document as a sequence of sentences, we define the document aggregate representation as
\begin{equation}
A_T
=
\sum_{t=1}^{T}x_t.
\end{equation}

The linearity allows the sum of sentence-level field values to equal the projection of this aggregate:
\begin{equation}
\begin{aligned}
S_K(D)
&=\sum_{t=1}^{T}f_K(x_t)
=\sum_{t=1}^{T}\omega_K^{\mathsf T}x_t
=\omega_K^{\mathsf T}A_T.
\end{aligned}
\end{equation}

Assuming $A_T\neq 0$, we normalize the aggregate as
\begin{equation}
u_T
=
\frac{A_T}{\lVert A_T\rVert_2},
\end{equation}
and define the document-level statistic
\begin{equation}
Z_K
=
\frac{\omega_K^{\mathsf T}A_T}{\lVert A_T\rVert_2}
=
\omega_K^{\mathsf T}u_T.
\end{equation}

The normalization by $\lVert A_T\rVert_2$ prevents repeated or highly correlated sentences from being counted as independent evidence. For example, if $x_1=\cdots=x_T=x$, then $A_T=Tx$, but 
\begin{equation}
Z_K
=
\frac{T\omega_K^{\mathsf T}x}{T}
=
\omega_K^{\mathsf T}x.
\end{equation}
Thus, repeating the sentence does not artificially inflate the watermark score.

\subsection{\semfield Generation}
\label{sec:sf-gen}

Generation proceeds sequentially, sentence by sentence. Suppose that $t-1$ sentences have already been selected, and their running aggregate is
\begin{equation}
A_{t-1}
=
\sum_{j=1}^{t-1}x_j.
\end{equation}

At step $t$, the language model $\mathcal{M}_{\phi}$ samples $Q$ candidate sentences
conditioned on the current document state $s_{<t}$:
\begin{equation}
c_{t,q}\sim\mathcal{M}_{\phi}(\,\cdot\mid s_{<t}),
\qquad q=1,\ldots,Q,
\end{equation}
forming the candidate set
\begin{equation}
C_t=\{c_{t,1},\ldots,c_{t,Q}\}.
\end{equation}

If candidate $q$ were selected, the new aggregate would be
\begin{equation}
A_{t,q}
=
A_{t-1}+x_{t,q}
\end{equation}

where $x_{t,q}$ is its normalized embedding. We score the \textit{prospective} aggregate with the document statistic used at detection:
\begin{equation}
G_{t,q}
=
\frac{\omega_K^{\mathsf T}A_{t,q}}{\lVert A_{t,q}\rVert_2}.
\end{equation}
\semfield then chooses the candidate with the largest score,
\begin{equation}
q_t^\star
=
\operatorname*{arg\,max}_{1\leq q\leq Q}G_{t,q},
\qquad t=1,\ldots,T,
\end{equation}
and sets $A_t=A_{t,q_t^\star}$. This \textit{greedy} rule favors alignment with the keyed direction at every step, and maximizes the exact score we use for detection as well.

\subsection{\semfield Detection}
\label{sec:detect-and-exact-cond-calib}

For a completed document $D$, the detector recomputes the sentence embeddings and their aggregate $A_T$. From that observed score is
\begin{equation}
z_{\mathrm{obs}}
=
\frac{\omega_K^{\mathsf T}A_T}{\lVert A_T\rVert_2}.
\end{equation}
A larger value of $z_{\mathrm{obs}}$ provides stronger watermark evidence
through the alignment of the document aggregate with $\omega_K$.

Under $H_0$, if the document is independent of the tested key, then by
\Cref{thm:conditional-gaussian-null},
\begin{equation}
Z_K\mid D\sim\mathcal{N}(0,1).
\end{equation}

The upper-tail probability of the observed score is
\begin{equation}
p_{\mathrm{SF}}
=
\Pr_{H_0}\!\left(Z_K\geq z_{\mathrm{obs}}\mid D\right)
=
1-\Phi(z_{\mathrm{obs}}),
\end{equation}
where $\Phi$ is the standard Normal CDF. A small $p_{\mathrm{SF}}$ indicates that the observed alignment is unusual under $H_0$, and that we have to reject it.

\subsection{Polarity Locking (PL)}
\label{sec:polarity-locking}

By definition, \semfield always selects sentences aligned towards $+\omega_K$. But the Gaussian construction is symmetric as $\omega_K$ and $-\omega_K$ have the same distribution, and neither orientation is intrinsically preferable before generating the candidates. Further, the Gaussian null and aggregate-based results discussed in prior sections also apply to either orientation. 

To finalize one orientation before full generation, we use the first candidate set. At $t=1$, it selects the candidate with the largest absolute score,
\begin{equation}
q_1^\star
=
\operatorname*{arg\,max}_{1\leq q\leq Q}|G_{1,q}|
\end{equation}
and records
\begin{equation}
\sigma
=
\operatorname{sign}\bigl(G_{1,q_1^\star}\bigr)
\in\{-1,+1\}.
\end{equation}
If the selected score is exactly zero, ties are resolved by setting $\sigma=+1$. For every later step, the candidate is selected such that it best improves the score along the fixed direction.
\begin{equation}
q_t^\star
=
\operatorname*{arg\,max}_{1\leq q\leq Q}\sigma G_{t,q}.
\end{equation}
The rule reduces to regular \semfield when $\sigma=+1$ and $\sigma=-1$ reinforces the negative direction. A fixed sign also avoids switching orientations between sentences and canceling their contributions in the final aggregate. We call this method polarity-locked \semfield, or \semfieldpl.

During detection, a document generated toward $+\omega_K$ tends to have a large positive projection, while one generated toward $-\omega_K$ tends to be largely negative. Since the detector does not know which orientation was selected, it treats either outcome as evidence. Under the symmetric Gaussian null, these events correspond to the \textit{disjoint} upper and lower tails,
respectively. Hence, writing $h_{\mathrm{obs}}=|z_{\mathrm{obs}}|$, their probabilities sum to
\begin{equation}
\begin{aligned}
p_{\mathrm{PL}}
&=
\Pr_{H_0}(Z_K\geq h_{\mathrm{obs}}\mid D)
+
\Pr_{H_0}(Z_K\leq-h_{\mathrm{obs}}\mid D)\\
&=
\left[1-\Phi\!\left(h_{\mathrm{obs}}\right)\right]
+
\Phi\!\left(-h_{\mathrm{obs}}\right)\\
&=
2\left[1-\Phi\!\left(h_{\mathrm{obs}}\right)\right]
\end{aligned}
\end{equation}

Once again, a small $p_{\mathrm{PL}}$ indicates that the observed alignment is unusual under $H_0$, and that we have to reject it.

\subsection{Structural Robustness}
\label{sec:struct-rob}

A structural attack changes a document by adding sentences, deleting existing sentences, or reordering them. Let $A$ and $A'$ denote the original and attacked document aggregates, respectively, and write
\begin{equation}
A'=A+\Delta,
\end{equation}
where $\Delta$ is the net semantic contribution of the attack.

For insertion or deletion attacks, provided $A,A'\neq 0$, the normalized aggregates satisfy
\begin{equation}
\left\lVert
\frac{A'}{\lVert A'\rVert_2}
-
\frac{A}{\lVert A\rVert_2}
\right\rVert_2
\leq
\frac{2\lVert\Delta\rVert_2}{\min\{\lVert A\rVert_2,\lVert A'\rVert_2\}}.
\end{equation}
Consequently, the change in the signed \semfield score is bounded by
\begin{equation}
|Z'_K-Z_K|
\leq
\frac{2\lVert\omega_K\rVert_2\lVert\Delta\rVert_2}{\min\{\lVert A\rVert_2,\lVert A'\rVert_2\}}.
\end{equation}
By the reverse triangle inequality, the same bound holds for the change in the \semfieldpl score $\bigl||Z'_K|-|Z_K|\bigr|$. Therefore, if a document's score exceeds the relevant detection threshold by more than this bound, the edit cannot change its positive decision. The complete proof of it is given in \Cref{prop:aggregate-stability}.

For sentence reordering, let $\pi$ be a permutation of $\{1,\ldots,T\}$. If sentence boundaries and individual sentence embeddings remain unchanged, then
\begin{equation}
A_T^{\pi}
=
\sum_{t=1}^{T}x_{\pi(t)}
=
\sum_{t=1}^{T}x_t
=
A_T.
\end{equation}
Similarly, $Z_K^{\pi}=Z_K$ and $|Z_K^{\pi}|=|Z_K|$. Therefore, detection by either \semfield or \semfieldpl is exactly invariant to reordering intact sentences. The formal proof is provided in \Cref{prop:reordering-invariance}.

\section{Experiments}

\subsection{Setup and Baselines}
Following previous works, we evaluate all methods on the same 500-document collection drawn from the \realnews subset. Each record contains an original human-written document, and its first sentence serves as the generation prompt, while the remainder serves as the human-written (unwatermarked) reference continuation. To accurately ascertain the efficacy of the \semfield variants, we run our experiments independently across 3 different language models: \llama \citep{grattafiori2024llama3herdmodels}, \qwen \citep{yang2025qwen3technicalreport}, and \opt \citep{zhang2022opt}. 

\paragraph{Watermark methods.}
We compare \semfield and \semfieldpl against 12 recent and relevant baselines for language modeling, \textit{viz.} \KGW, \EXP, \synthid, \SIR, \semstamp, \ksemstamp, \simmark, \pmark, \alimark, \samark, \swordstamp, and \kswordstamp. For each baseline, any algorithm-specific parameters required by the original implementation are kept unchanged.

\paragraph{Generation protocol.}
For each document \(i\), let \(p_i\) denote its initial prompt. A clean watermarked output is generated by continuing \(p_i\) with the selected watermark method:
\[
    y_i
    \sim
    P_{\mathcal{M}_{\phi}}\!\left(
        y \mid p_i,\text{watermark method},K
    \right),
\]
where $\mathcal{M}_{\phi}$ is the language model and $K$ is the method’s secret key. For methods that use best-of-$Q$ candidate selection, we set $Q=64$, matching the candidate budget used by \citet{li2026alimark}, and the semantic algorithms use \mpnet as the embedding model. All models and tested algorithms use $\operatorname{temperature}=0.7$, $\operatorname{top\_p}=0.95$, $\operatorname{repetition\_penalty}=1.05$ during generation.

\subsection{Attacks and Evaluation}
The primary objective of a watermark is reliable detection after the text has been modified. Hence, we study the robustness of our method across various types of attacks and textual modification strategies. For all experiments, we first measure clean detection on nascent/unattacked watermarked outputs.

\paragraph{Paraphrasing.} 
A standard process that LLM-generated outputs are subjected to is paraphrasing, i.e., rewriting the textual content of a document while retaining its meaning as much as possible. Hence, we then evaluate robustness to rewriting at two levels. All attacks receive only the clean watermarked text and do not have access to the secret key or detector.

\begin{enumerate}
    \item \textbf{Sentence-level rewriting:} This attack treats the document as a list of sentences and independently paraphrases each sentence. Hence, the order of the sentences and the information they convey remains the same, but they are reworded. \pegasus\footnote{\url{https://huggingface.co/tuner007/pegasus_paraphrase}} \citep{zhang2020pegasus} and \parrot \citep{prithivida2021parrot} are two such attacks that change
    local wording while approximately preserving the original sentence boundaries and document structure.

    \item \textbf{Document-level rewriting:} A stronger and more comprehensive attack would be to completely paraphrase the textual content end-to-end. Such a transformation destroys all sentence structure and is aimed at only preserving the overall meaning of the document, whose detection is a true test of a semantic watermark. Hence, we treat the following as the primary attacks for our work: (i) \dipper \citep{krishna2023paraphrasing} is a \textsc{T5}-based paraphraser model configured to encourage lexical\footnote{how many words should change} and order\footnote{how much the sentence structure should change.} diversity while preserving the document meaning, (ii) \llmrewrite uses \qwenrewriter \citep{qwen38}, with the instruction to rewrite the entire document, avoid reusing the original words and phrases where possible, and preserve the original meaning.
\end{enumerate}

\paragraph{Detection metrics.}
Each detector produces a scalar score, with higher values indicating stronger evidence for a watermark. For each method and attack, we report:

\begin{enumerate}
    \item \auroc: Measures how well the detector separates watermarked text from unwatermarked text across all possible detection thresholds. The ROC curve plots \tpr against \fpr at various thresholds, and \auroc is the area under that curve. Formally, $\mathrm{AUROC} = \Pr\!\left(S_{\text{wm}}>S_{\text{null}}\right) + \frac{1}{2}\Pr\!\left(S_{\text{wm}} = S_{\text{null}}\right)$,
    and describes the probability that a randomly chosen watermarked document receives a higher score than a randomly chosen null document, with half credit for ties. An \auroc of $1.0$ indicates perfect separation, while $0.5$ is a random guess.

    \item $\operatorname{TPR}@p\%\operatorname{FPR}$: Measures the fraction of watermarked documents detected at the target $p\%$ \fpr. For threshold $\tau_p$, the target is $\Pr\!\left(S_{\mathrm{null}} \geq \tau_p\right)=p/100$. The reported detection power is then
    $\operatorname{TPR}@p\%\operatorname{FPR}=\Pr\!\left(S_{\mathrm{wm}} \geq \tau_p\right)$. In this work, we report $p \in \{1,5\}$, with the former being the stricter and main evaluation point. We estimate these operating points from the empirical ROC curve, using interpolation at the target \fpr when needed. An extension is provided in the following analysis, where in the \fpr cut-offs are empirically estimated from a large set of null documents \textit{a priori} for a deployable scenario.
\end{enumerate}

We also report the mean \tpro\fpr and \tprf\fpr across the clean
condition and four rewrite attacks. This summarizes detection performance across both unattacked and rewritten text, and provides a more \textit{holistic} evaluation of the algorithms,
\begin{equation}
    \operatorname{MeanTPR}
    =
    \frac{1}{|\mathcal{A}|}
    \sum_{a\in\mathcal{A}}
    \operatorname{TPR}_{a},
\end{equation}
where
\[
    \mathcal{A}
    =
    \{\clean,\pegasus,\parrot,\dipper,\llmrewrite\}.
\]

\paragraph{Text quality.}

To measure whether watermarking makes the output less natural, we evaluate the clean watermarked text with \llamajudge \citep{grattafiori2024llama3herdmodels} as a language modeling judge for all methods. To this end, we report corpus-level perplexity, and the $\operatorname{PPLRatio}$, which we define as:
\begin{equation}
    \operatorname{PPLRatio}
    =
    \frac{
        \operatorname{PPL}(\text{watermarked text})
    }{
        \operatorname{PPL}(\text{human reference})
    }.
\end{equation}
A ratio near 1.0 indicates that the generated text has a predictability similar to the human reference, while a larger ratio indicates a larger deviation from the reference distribution.

\paragraph{Structural stress tests.}

In a separate experiment, we evaluate insertion, dilution, deletion, and sentence reordering attacks at several attack strengths \citep{li2026alimark}. Note that, in these attacks, no paraphrasing of the text is done. These tests measure \tpro\fpr against the fraction of inserted, deleted, or reordered sentences. We describe each attack below:

\begin{enumerate}
    \item \textbf{Insertion:}
    We insert short, semantically empty sentence markers ($.$, $!$, $?$ etc.) at randomly selected sentence boundaries. The attack strength $\delta$ denotes the fraction of original sentence positions at which an additional marker is inserted. This tests whether the detector is affected by extra sentence boundaries that carry no watermark information.

    \item \textbf{Dilution:}
    We insert sentences sampled from an independent null-text pool into the watermarked document. The original sentences are retained, while the added sentences contribute no watermark signal. At strength $\delta$, the number of inserted null sentences is chosen so that they form a fraction $\delta$ of the resulting document. This measures how well the detector retains evidence when watermark-bearing content is mixed with unrelated text.

    \item \textbf{Deletion:}
    We randomly remove a fraction $\delta$ of the original sentences. The remaining sentences preserve their original order, but some watermark evidence is lost. This evaluates the detector's tolerance to missing content.

    \item \textbf{Reordering:}
    We randomly select a fraction $\delta$ of the sentences and randomly permute their positions while leaving the sentence content unchanged. This attack changes the document order without changing the individual sentence meanings or embeddings. It therefore tests whether detection depends on sentence order.
\end{enumerate}

\section{Results and Analysis}

In this section, we discuss the primary results of our algorithm as compared to all the baselines, along with some observations and ablations. Any score triples reported here on are in the order of \llama, \qwen, and \opt, respectively.

\subsection{Overall Detection, Paraphrase Robustness, and Textual Quality}

\paragraph{Ours \textit{vs.} Baselines.}
\Cref{tab:llama32-1b-main-no-nonce,tab:qwen3-1.7b-main-no-nonce,tab:opt-1.3b-main-no-nonce} present the main results for \llama, \qwen, and \opt, respectively. Across all three models, \semfield and \semfieldpl occupy the top two positions in mean \tpro\fpr. Their mean scores range from $88.2\%$ to $90.4\%$, compared with $78.0\%$, $78.2\%$, and $81.7\%$ for the strongest baseline. \semfieldpl also obtains the highest mean \tprf\fpr of $96.6\%$, $96.4\%$, and $96.4\%$, while \semfield has the second best on \llama, and a minor difference of $-0.2$, and $-0.4$ with respect to the second best on \qwen and \opt.

The strongest paraphrasing attack is \llmrewrite, which rigorously alters the input text end-to-end. Even there, both \semfield variants exceed the best baseline in \tpro\fpr on all three models. For example, \semfieldpl obtains $66.4\%$, $76.0\%$, and $69.2\%$, whereas the strongest baseline, \samark, obtains $64.6\%$, $60.2\%$, and $67.2\%$, with the highest gain of $+15.8$ on \qwen. Across the remaining sentence-level attacks, \pegasus and \parrot, each \semfield variant exceeds the strongest baseline in five of the six model--attack comparisons. More specifically, \semfieldpl has a margin of $2.6$--$8.6\%$ on \pegasus, while on \parrot, it shows $6.6\%$ and $5.0\%$ gain for \llama and \qwen, respectively. The only exception is \opt under \parrot, where \swordstamp scores meager $0.4\%$ higher than \semfieldpl at $94.6\%$.

\begin{table*}[h]
\centering
\scriptsize
\setlength{\tabcolsep}{2.0pt}
\renewcommand{\arraystretch}{1.10}
\caption{Detection results for \llama. Best values in a column are bolded, and the second best are underlined.}
\label{tab:llama32-1b-main-no-nonce}
\resizebox{\textwidth}{!}{%
\begin{tabular}{lccccc|cc|c}
\toprule
Method & \clean & \pegasus & \parrot & \dipper & \llmrewrite & \shortstack{Mean\\\tpro} & \shortstack{Mean\\\tprf} & \shortstack{Corp. PPL $\downarrow$\\(\pplratio)} \\
\midrule
\rowcolor{gray!15}
\multicolumn{9}{@{}l}{\textsc{Token-based watermarks}} \\
\addlinespace[2pt]
\KGW & \best{100.0} / \second{99.8} / \best{100.0} & 97.8 / 84.4 / 90.8 & 98.6 / 89.2 / 93.4 & 93.0 / 57.2 / 69.0 & 70.0 / 13.4 / 22.4 & 68.8 & 75.1 & 8.7 (0.897) \\
\synthid & \best{100.0} / \best{100.0} / \best{100.0} & 91.2 / 48.6 / 72.6 & 90.6 / 50.8 / 72.0 & 73.5 / 12.6 / 31.6 & 57.2 / 4.6 / 16.0 & 43.3 & 58.4 & 7.7 (0.793) \\
\EXP & 99.8 / 99.2 / \second{99.4} & 89.5 / 59.0 / 69.2 & 89.2 / 61.0 / 71.4 & 80.0 / 31.4 / 45.4 & 56.2 / 5.0 / 12.8 & 51.1 & 59.6 & \second{4.5} (\second{0.471}) \\
\SIR & \second{99.9} / 98.8 / \second{99.4} & 95.1 / 82.4 / 88.6 & 96.2 / 81.8 / 87.4 & 93.0 / 71.2 / 79.4 & 82.9 / 27.2 / 44.2 & 72.3 & 79.8 & 7.6 (0.783) \\
\addlinespace[3pt]
\rowcolor{blue!10}
\multicolumn{9}{@{}l}{\textsc{Semantic watermarks}} \\
\addlinespace[2pt]
\semstamp & \second{99.9} / 99.6 / \best{100.0} & 96.6 / 79.2 / 89.6 & 97.5 / 82.2 / 92.6 & 87.0 / 41.6 / 61.6 & 66.7 / 15.8 / 28.2 & 63.7 & 74.4 & 5.7 (0.591) \\
\ksemstamp & 99.4 / 96.5 / 99.2 & 94.7 / 46.0 / 75.4 & 96.6 / 53.2 / 83.5 & 85.9 / 18.1 / 47.7 & 71.7 / 4.4 / 16.8 & 43.7 & 64.5 & 6.1 (0.633) \\
\simmark & 99.2 / 93.0 / 97.0 & 98.2 / 84.6 / 93.0 & 98.7 / 88.2 / 94.4 & 95.3 / 60.8 / 81.6 & 84.2 / 22.4 / 45.2 & 69.8 & 82.2 & 6.1 (0.635) \\
\pmark & 99.4 / 96.2 / 98.2 & 95.6 / 74.8 / 86.6 & 98.0 / 86.4 / 92.4 & 87.1 / 48.6 / 64.6 & 74.6 / 14.2 / 31.6 & 64.0 & 74.7 & \best{3.9} (\best{0.401}) \\
\alimark & 99.8 / 98.0 / \second{99.4} & 97.9 / 80.6 / 90.4 & 98.3 / 85.4 / 93.3 & 92.5 / 58.4 / 71.8 & 89.2 / 37.2 / 57.3 & 71.9 & 82.4 & 11.3 (1.168) \\
\samark & 95.2 / 79.8 / 87.4 & 95.3 / 76.0 / 86.4 & 96.4 / 82.4 / 90.0 & 93.3 / 66.0 / 79.4 & 92.1 / 64.6 / 77.0 & 73.8 & 84.0 & 6.8 (0.708) \\
\swordstamp & 99.7 / 93.2 / 98.3 & 99.1 / 86.8 / 94.2 & 99.3 / 87.4 / 95.8 & 96.4 / 72.2 / 84.3 & 84.4 / 22.8 / 44.9 & 72.5 & 83.5 & 6.7 (0.691) \\
\kswordstamp & 99.7 / 92.6 / 98.6 & 99.0 / 84.0 / 95.6 & 99.2 / 85.4 / 96.2 & 97.7 / 72.4 / 89.6 & 95.5 / 55.6 / \second{80.8} & 78.0 & 92.2 & 6.8 (0.707) \\
\addlinespace[3pt]
\rowcolor{green!12}
\multicolumn{9}{@{}l}{\textsc{Ours}} \\
\addlinespace[2pt]
\semfield & 99.8 / 97.0 / 98.8 & \second{99.4} / \second{93.6} / \second{96.8} & \second{99.5} / \second{93.8} / \second{97.6} & \second{98.9} / \best{89.0} / \second{94.0} & \second{96.2} / \best{68.4} / \second{80.8} & \best{88.4} & \second{93.6} & 6.8 (0.703) \\
\semfieldpl & \best{100.0} / 99.0 / \best{100.0} & \best{99.7} / \best{94.2} / \best{99.0} & \best{99.8} / \best{95.8} / \best{99.2} & \best{99.1} / \second{85.4} / \best{95.8} & \best{97.8} / \second{66.4} / \best{89.2} & \second{88.2} & \best{96.6} & 6.9 (0.715) \\
\bottomrule
\end{tabular}%
}
\end{table*}

\begin{table*}[h]
\centering
\scriptsize
\setlength{\tabcolsep}{2.0pt}
\renewcommand{\arraystretch}{1.10}
\caption{Detection results for \qwen. Best values in a column are bolded, and the second best are underlined.}
\label{tab:qwen3-1.7b-main-no-nonce}
\resizebox{\textwidth}{!}{%
\begin{tabular}{lccccc|cc|c}
\toprule
Method & \clean & \pegasus & \parrot & \dipper & \llmrewrite & \shortstack{Mean\\\tpro} & \shortstack{Mean\\\tprf} & \shortstack{Corp. PPL $\downarrow$\\(\pplratio)} \\
\midrule
\rowcolor{gray!15}
\multicolumn{9}{@{}l}{\textsc{Token-based watermarks}} \\
\addlinespace[2pt]
\KGW & \second{99.9} / \second{99.6} / \second{99.8} & 96.8 / 80.4 / 90.2 & 98.8 / 88.4 / 96.2 & 91.4 / 45.4 / 69.6 & 76.2 / 10.6 / 28.6 & 64.9 & 76.9 & 10.9 (1.125) \\
\synthid & \best{100.0} / \best{99.8} / \second{99.8} & 90.3 / 59.2 / 73.8 & 94.0 / 70.6 / 80.8 & 72.5 / 19.2 / 32.8 & 59.3 / 11.2 / 19.0 & 52.0 & 61.2 & 9.7 (1.007) \\
\EXP & \second{99.9} / \best{99.8} / \second{99.8} & 91.3 / 62.4 / 73.8 & 92.7 / 66.6 / 79.2 & 78.9 / 34.0 / 51.4 & 59.9 / 7.4 / 15.4 & 54.0 & 63.9 & \second{5.2} (\second{0.535}) \\
\SIR & \second{99.9} / 98.8 / 99.6 & 95.6 / 86.0 / 90.4 & 98.4 / 90.8 / 94.2 & 95.2 / 69.0 / 82.2 & 81.7 / 26.8 / 41.4 & 74.3 & 81.6 & 10.4 (1.075) \\
\addlinespace[3pt]
\rowcolor{blue!10}
\multicolumn{9}{@{}l}{\textsc{Semantic watermarks}} \\
\addlinespace[2pt]
\semstamp & 99.6 / 98.2 / 99.0 & 94.5 / 72.4 / 83.8 & 96.5 / 80.0 / 89.8 & 80.8 / 31.0 / 48.0 & 64.5 / 8.8 / 24.0 & 58.1 & 68.9 & 8.5 (0.881) \\
\ksemstamp & 98.9 / 91.6 / 98.7 & 93.6 / 37.7 / 68.5 & 94.4 / 45.4 / 76.4 & 83.3 / 12.8 / 37.3 & 71.3 / 3.6 / 20.9 & 38.2 & 60.4 & 7.6 (0.792) \\
\simmark & 98.6 / 87.0 / 94.8 & 97.8 / 78.0 / 89.6 & 98.1 / 79.6 / 93.0 & 93.8 / 47.4 / 73.2 & 85.7 / 25.4 / 53.4 & 63.5 & 80.8 & 9.0 (0.930) \\
\pmark & 99.5 / 98.6 / 98.8 & 96.3 / 81.2 / 86.6 & 97.8 / 88.2 / 92.8 & 86.6 / 52.2 / 62.8 & 77.0 / 22.6 / 32.4 & 68.6 & 74.7 & \best{4.3} (\best{0.443}) \\
\alimark & \best{100.0} / 99.4 / \best{100.0} & 97.4 / 80.6 / 90.9 & 98.3 / 88.4 / 94.4 & 89.6 / 47.4 / 64.3 & 88.1 / 37.0 / 57.3 & 70.6 & 81.4 & 19.0 (1.970) \\
\samark & 96.9 / 86.2 / 93.2 & 96.8 / 81.9 / 90.0 & 97.6 / 83.7 / 91.4 & 95.6 / 69.5 / 85.3 & 93.9 / 60.2 / 78.7 & 76.3 & 87.7 & 8.5 (0.881) \\
\swordstamp & 99.8 / 96.6 / 99.0 & 99.1 / 85.7 / 95.9 & \second{99.5} / 90.0 / \second{97.6} & 96.9 / 65.2 / 86.5 & 85.2 / 18.0 / 43.1 & 71.1 & 84.4 & 9.3 (0.961) \\
\kswordstamp & 99.8 / 94.8 / 99.2 & 98.9 / 79.5 / 94.8 & 99.4 / 86.3 / 96.2 & 97.9 / 74.5 / \second{93.3} & \second{96.9} / 56.0 / \second{85.5} & 78.2 & \second{93.8} & 9.0 (0.936) \\
\addlinespace[3pt]
\rowcolor{green!12}
\multicolumn{9}{@{}l}{\textsc{Ours}} \\
\addlinespace[2pt]
\semfield & 99.8 / 97.2 / 98.0 & \second{99.4} / \second{93.0} / \second{97.2} & \second{99.5} / \second{93.8} / 97.2 & \second{98.8} / \best{87.6} / 92.8 & 96.6 / \second{71.4} / 82.6 & \second{88.6} & 93.6 & 8.6 (0.889) \\
\semfieldpl & \second{99.9} / \second{99.6} / \second{99.8} & \best{99.7} / \best{94.6} / \best{98.8} & \best{99.8} / \best{95.8} / \best{99.4} & \best{98.9} / \second{85.8} / \best{95.0} & \best{97.9} / \best{76.0} / \best{89.2} & \best{90.4} & \best{96.4} & 8.8 (0.910) \\
\bottomrule
\end{tabular}%
}
\end{table*}

\begin{table*}[h]
\centering
\scriptsize
\setlength{\tabcolsep}{2.0pt}
\renewcommand{\arraystretch}{1.10}
\caption{Detection results for \opt. Best values in a column are bolded, and the second best are underlined.}
\label{tab:opt-1.3b-main-no-nonce}
\resizebox{\textwidth}{!}{%
\begin{tabular}{lccccc|cc|c}
\toprule
Method & \clean & \pegasus & \parrot & \dipper & \llmrewrite & \shortstack{Mean\\\tpro} & \shortstack{Mean\\\tprf} & \shortstack{Corp. PPL $\downarrow$\\(\pplratio)} \\
\midrule
\rowcolor{gray!15}
\multicolumn{9}{@{}l}{\textsc{Token-based watermarks}} \\
\addlinespace[2pt]
\KGW & 99.8 / \second{99.8} / \second{99.8} & 97.6 / 72.4 / 90.2 & 98.1 / 80.4 / 93.6 & 91.2 / 35.0 / 62.2 & 68.1 / 4.2 / 16.6 & 58.4 & 72.5 & 8.5 (0.876) \\
\synthid & \best{100.0} / \best{100.0} / \best{100.0} & 90.0 / 53.8 / 69.4 & 90.4 / 58.8 / 72.8 & 76.1 / 22.0 / 38.0 & 58.3 / 7.8 / 16.4 & 48.5 & 59.3 & 8.0 (0.832) \\
\EXP & 99.7 / 98.8 / 99.6 & 88.6 / 48.7 / 68.5 & 89.0 / 52.7 / 70.3 & 76.2 / 21.6 / 39.3 & 55.5 / 3.8 / 11.4 & 45.1 & 57.8 & \second{5.2} (\second{0.541}) \\
\SIR & 99.5 / 98.6 / 98.8 & 96.4 / 85.8 / 91.0 & 97.4 / 86.0 / 92.2 & 91.5 / 58.8 / 74.0 & 85.2 / 25.0 / 47.8 & 70.8 & 80.8 & 8.3 (0.856) \\
\addlinespace[3pt]
\rowcolor{blue!10}
\multicolumn{9}{@{}l}{\textsc{Semantic watermarks}} \\
\addlinespace[2pt]
\semstamp & \second{99.9} / 99.6 / \best{100.0} & 95.1 / 74.2 / 86.4 & 96.8 / 81.0 / 89.6 & 86.1 / 38.0 / 59.0 & 66.4 / 13.2 / 24.0 & 61.2 & 71.8 & 8.5 (0.884) \\
\ksemstamp & 99.7 / 92.6 / \second{99.8} & 94.7 / 42.4 / 75.5 & 95.4 / 41.7 / 77.1 & 87.0 / 14.9 / 47.7 & 69.6 / 3.4 / 14.8 & 39.0 & 63.0 & 6.5 (0.668) \\
\simmark & 98.6 / 89.4 / 94.2 & 97.1 / 77.8 / 89.4 & 97.5 / 79.8 / 91.2 & 93.8 / 55.8 / 75.8 & 86.3 / 30.2 / 52.8 & 66.6 & 80.7 & 6.7 (0.698) \\
\pmark & \second{99.9} / 99.4 / \second{99.8} & 95.7 / 76.4 / 86.0 & 99.1 / 93.2 / 97.0 & 86.2 / 48.6 / 63.0 & 69.3 / 13.6 / 26.8 & 66.2 & 74.5 & \best{4.3} (\best{0.442}) \\
\alimark & 99.8 / 99.4 / 99.5 & 98.5 / 87.6 / 94.9 & 98.8 / 93.4 / 96.0 & 92.6 / 55.0 / 71.8 & 91.1 / 48.0 / 65.8 & 76.7 & 85.6 & 8.5 (0.878) \\
\samark & 98.9 / 90.2 / 97.4 & 98.7 / 86.8 / 94.6 & 99.1 / 91.6 / 97.0 & 96.8 / 72.8 / 87.6 & 95.4 / 67.2 / 82.2 & 81.7 & 91.8 & 6.7 (0.697) \\
\swordstamp & \best{100.0} / \second{99.8} / \second{99.8} & \second{99.5} / 92.0 / \second{97.6} & \second{99.7} / \best{95.0} / \second{98.8} & 97.5 / 76.6 / 87.8 & 86.5 / 20.8 / 39.9 & 76.8 & 84.8 & 6.7 (0.689) \\
\kswordstamp & \second{99.9} / 96.0 / \best{100.0} & 99.1 / 81.8 / 96.6 & 99.5 / 89.0 / 98.2 & 98.2 / 73.6 / 92.4 & 96.1 / 49.6 / \second{83.7} & 78.0 & \second{94.2} & 7.0 (0.725) \\
\addlinespace[3pt]
\rowcolor{green!12}
\multicolumn{9}{@{}l}{\textsc{Ours}} \\
\addlinespace[2pt]
\semfield & 99.8 / 97.4 / 99.0 & \second{99.5} / \second{92.6} / 96.8 & 99.5 / 93.2 / 97.2 & \second{99.0} / \best{88.8} / \second{93.8} & \second{96.6} / \best{69.4} / 82.0 & \best{88.3} & 93.8 & 6.3 (0.656) \\
\semfieldpl & \best{100.0} / \second{99.8} / \best{100.0} & \best{99.8} / \best{94.6} / \best{99.6} & \best{99.8} / \second{94.6} / \best{99.4} & \best{99.1} / \second{82.6} / \best{95.4} & \best{97.5} / \second{69.2} / \best{87.6} & \second{88.2} & \best{96.4} & 6.4 (0.661) \\
\bottomrule
\end{tabular}%
}
\end{table*}

\begin{wraptable}{r}{0.48\textwidth}
\centering
\small
\setlength{\tabcolsep}{4pt}
\caption{Corpus perplexity compared with the median across 12 baselines. PPL column reports \semfield / \semfieldpl}
\label{tab:semfield-baseline-median-ppl}
\begin{tabular}{lcc}
\toprule
Model & Baseline median & Corp. PPL $\downarrow$ \\
\midrule
\llama & 6.75 & 6.8 / 6.9 \\
\qwen  & 9.00 & 8.6 / 8.8 \\
\opt   & 6.85 & 6.3 / 6.4 \\
\bottomrule
\end{tabular}
\end{wraptable}

Lastly, \semfield and \semfieldpl achieve substantially stronger detection without moving into a high-perplexity regime. As seen from \Cref{tab:semfield-baseline-median-ppl}, both variants have Corpus PPL close to the baseline median on \llama and below it on \qwen and \opt, placing them within the same general perplexity range as the methods we compare against.

\paragraph{Effect of Polarity-Locking (PL).}
PL allows generation to reinforce either orientation of the keyed field, while \semfield always reinforces the positive orientation. This change improves \auroc and \tprf\fpr in every one of the five conditions on all three models. Averaged over the five conditions, PL raises \tprf\fpr by $3.0$, $2.8$, and $2.6\%$ compared to its one-sided variant. It also improves \tpro\fpr in ten of the fifteen model--condition comparisons. Interestingly, the strict $1\%$ operating point does not favor polarity-locking in every condition in our current results. 

\begin{wraptable}{r}{0.48\textwidth}
\centering
\small
\setlength{\tabcolsep}{4pt}
\caption{Mean document scores before and after \llmrewrite. $\text{Score loss} = \text{Clean} - \text{Rewritten}$. Each cell reports \semfield / \semfieldpl. Bold marks the higher score or, for score loss, the lower loss.}
\label{tab:llm-rewrite-score-retention}
\resizebox{\linewidth}{!}{%
\begin{tabular}{lccc}
\toprule
Model & Clean & Rewritten & Score loss $\downarrow$ \\
\midrule
\llama & 4.475 / \textbf{4.632}
       & 3.042 / \textbf{3.221}
       & 1.433 / \textbf{1.411} \\
\qwen  & 4.540 / \textbf{4.726}
       & 3.194 / \textbf{3.380}
       & \textbf{1.346 / 1.346} \\
\opt   & 4.418 / \textbf{4.610}
       & 3.089 / \textbf{3.247}
       & \textbf{1.329} / 1.363 \\
\bottomrule
\end{tabular}
}%
\end{wraptable}

To understand this difference further, we first compare the mean document score before and after \llmrewrite for both variants. From \Cref{tab:llm-rewrite-score-retention}, we find that both variants display a very similar score drop under rewriting, indicating that \semfieldpl maintains comparable score stability under this attack. Next, analyzing the cutoff thresholds reveals the main difference. \semfield uses only the positive tail of the null scores to set its $1\%$ \fpr cutoff. \semfieldpl uses both the positive and negative tails, giving it a higher empirical cutoff of $\sim2.81$--$2.83$, compared with $\sim2.56$ for \semfield. This higher threshold leaves more borderline rewritten outputs below the PL detection cutoff on \llama and \opt, where its \tpro is lower by $2.0$ and $0.2\%$. On \qwen, more PL outputs clear the cutoff, giving it a $4.6\%$ lead. At $5\%$ \fpr, the empirical cutoffs are closer in this evaluation, and \semfieldpl achieves higher \llmrewrite \tpr on all three models by $8.4$, $6.6$, and $5.6\%$. Hence, we conclude that the drop in \tpr of PL under \llmrewrite was not due to its lack of robustness, but due to its stricter two-sided cutoff at $1\%$ \fpr.

Lastly, the perplexity cost of \semfieldpl is modestly higher than the former, with a difference of at most $+0.1$ in all cases. However, we reiterate that despite the minor difference, both techniques lie in proximity to the median of all the baselines, indicating no abnormal generation style. We therefore report both variants as regular \semfield establishes the strength of the continuous field itself, while \semfieldpl makes fuller use of its two possible orientations.

\paragraph{Additional Quality Evaluation.}
\begin{wraptable}{r}{0.48\textwidth}
\centering
\setlength{\tabcolsep}{3pt}
\caption{MAUVE scores against matched human-reference continuations.
Each value is the mean over three clustering seeds.}
\label{tab:mauve-all-models}
\tiny
\begin{tabular}{lccc}
\toprule
Method & \llama & \qwen & \opt \\
\midrule
\rowcolor{gray!15}
\multicolumn{4}{@{}l}{\textsc{Token-based watermarks}} \\
\addlinespace[2pt]
\KGW          & 0.020 & 0.017 & 0.014 \\
\synthid      & 0.013 & 0.017 & 0.023 \\
\EXP          & 0.007 & 0.008 & 0.012 \\
\SIR          & 0.029 & 0.011 & 0.016 \\
\addlinespace[3pt]
\rowcolor{blue!10}
\multicolumn{4}{@{}l}{\textsc{Semantic watermarks}} \\
\addlinespace[2pt]
\semstamp     & 0.240 & 0.332 & 0.152 \\
\ksemstamp    & 0.240 & 0.349 & 0.117 \\
\simmark      & 0.439 & 0.319 & 0.347 \\
\pmark        & 0.371 & 0.415 & 0.479 \\
\alimark      & 0.396 & 0.115 & 0.490 \\
\samark       & 0.472 & 0.341 & 0.559 \\
\swordstamp   & 0.302 & 0.293 & 0.083 \\
\kswordstamp  & 0.192 & 0.182 & 0.072 \\
\midrule
\rowcolor{green!12}
\multicolumn{4}{@{}l}{\textsc{Ours}} \\
\addlinespace[2pt]
\semfield     & \underline{0.800} & \underline{0.584} & \underline{0.699} \\
\semfieldpl   & \textbf{0.817} & \textbf{0.656} & \textbf{0.743} \\
\bottomrule
\end{tabular}
\end{wraptable}

We compare generated continuations with matched human-reference continuations using MAUVE \citep{NEURIPS2021_260c2432}. Since the clustering step of MAUVE is a stochastic process, we report the mean over three seeds. From \Cref{tab:mauve-all-models}, we find that \semfield and \semfieldpl obtain the two highest scores for each of the three models. \semfieldpl scores $0.817$, $0.656$, and $0.743$, compared with $0.472$, $0.415$, and $0.559$ for the strongest baseline on each model. The regular \semfield variant also exceeds every baseline, scoring $0.800$, $0.584$, and $0.699$. Thus, the strong detection results of both variants are accompanied by generated-text distributions that are closer to the human-reference continuations under this metric. Since MAUVE evaluates distributions rather than individual responses, we consider it as a secondary metric alongside perplexity when assessing generation quality.

\subsection{Structural Robustness}
Here, we mainly compare our methods with the three most recent and best competitors from our cohort -- \samark, \swordstamp, and \kswordstamp.

\Cref{fig:struct-rob} reports \tpro as attack strength increases from $0$ to $80\%$. Here, we mainly compare our methods with the three most recent and best competitors from our cohort -- \samark, \swordstamp, and \kswordstamp, as we find that most of the other baselines drop rapidly compared to our methods. \semfieldpl retains substantial detection under severe insertion and deletion. At $80\%$ insertion, it obtains $70.4$--$75.2\%$ \tpro across the three models, compared with $47.8$--$62.2\%$ for \semfield. At the same deletion strength, the respective ranges are $79.2$--$84.2\%$ and $71.8$--$80.6\%$. These results are consistent with the aggregate-perturbation analysis presented in \Cref{sec:struct-rob} as detection can persist when the edited aggregate retains enough of its keyed direction.

The considered baselines reveal different strengths. For instance, \swordstamp and \kswordstamp are particularly resistant to insertion. At $80\%$ strength, their \tpro remains at $91.5$--$98.6\%$ and $90.8$--$93.8\%$, respectively, close to their clean levels. However, their advantage does not extend to severe deletion. At $80\%$ deletion, \swordstamp falls to $14.6$--$27.8\%$ and \kswordstamp to $8.2$--$10.8\%$, while both \semfield variants retain substantially higher detection. \samark declines more sharply under insertion, reaching $0\%$ at $80\%$ strength, and retains $20.6$--$22.8\%$ under deletion.

For \semfield and \semfieldpl, permuting intact sentences leaves their aggregate and detector score exactly unchanged, as demonstrated in \Cref{sec:struct-rob}. Their reordering curves are therefore constant at their clean detection levels of $96.8$--$98.0\%$ and $99.6$--$100.0\%$, respectively. \samark, \swordstamp, and \kswordstamp also demonstrate flat reordering curves in this experiment owing to their design, but lie below both our variants across all 3 models.

Dilution best separates the two \semfield variants across all 4 attacks. At $80\%$ strength, \semfield retains $37.4$ -$52.8\%$ \tpro, compared with $21.4$--$28.6\%$ for \semfieldpl. At that strength, \swordstamp retains $23.6$--$28.0\%$, \kswordstamp retains $10.4$--$10.8\%$, and \samark falls to $0.2\%$. Following the aggregate-perturbation, both algorithms show good robustness till 40\% attack strength, after which the one-sided variant dominates and \semfieldpl drops sharply. Despite this, \semfieldpl still outperforms most baselines considered along the curve. Overall, in this setting, the one-sided \semfield variant is more resistant to that change, whereas polarity locking gives stronger insertion and generally stronger deletion performance.

Surprisingly, \simmark exhibits an unusual upward dilution curve, reaching $99.8$--$100.0\%$ \tpro at $80\%$ strength. Upon a deeper investigation, we find that its detector scores neighboring sentence pairs, and dilution creates additional pairs
that can raise the document score\footnote{analogous to \semfield's \(Z_K\)} used for detection. To appropriately highlight this effect, we also dilute the 500 null documents to check if the sentence pairs cause a similar gain on them. As expected, their mean \simmark score rises from $0.22$ to $12.39$ at $80\%$ strength, indicating that dilution raises scores even when no watermark is present. 

\begin{figure*}[t]
\centering
\includegraphics[width=\textwidth]{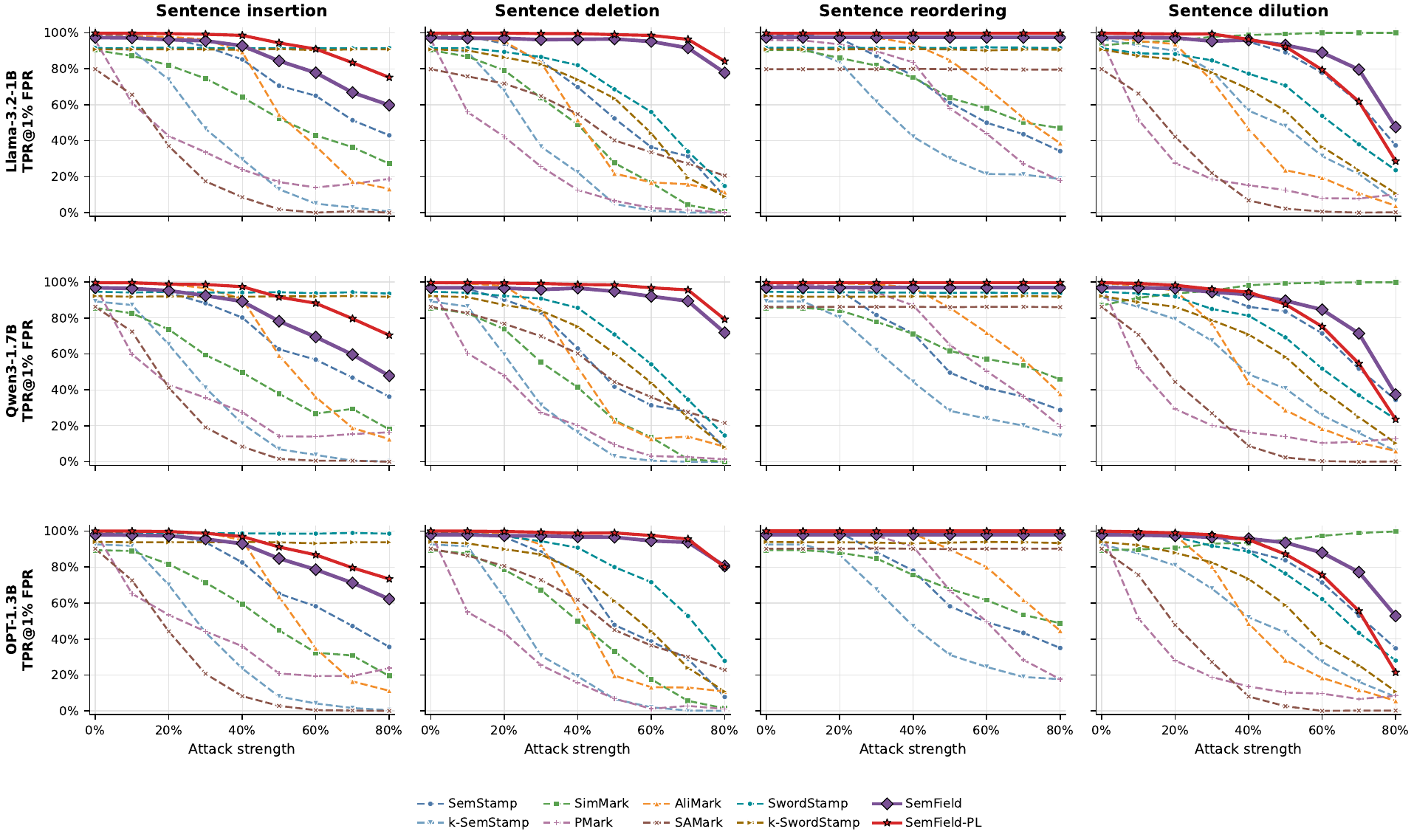}
\caption{Structural robustness under sentence insertion, deletion,
reordering, and dilution. Curves report \tpro\fpr.}
\label{fig:struct-rob}
\end{figure*}

\subsection{Effect of Candidate Diversity.}
\semfield variants operate on selecting that candidate that best aligns with the secret keyed vector and maximally improves the score of the resulting document. Hence, the diversity of candidates in a crucial as a collapsed set can result in low-quality additions to the document. We study the effect of candidate diversity and associated side-effects along two arms as follows. We limit this study to \semfieldpl as it shows a more \textit{broader} improvement across all models and conditions.

\paragraph{Temperature.}
Sampling temperature changes the candidate sentences available at a fixed search budget. A high temperature produces more unseen or rare generations by selecting tokens towards the tail of the distribution. \Cref{tab:semfield-temperature-ablation-no-nonce-roc500} compares temperatures $0.7$ and $1.0$ on the \llama evaluation with the rest of the generation parameters unchanged. At temperature $1.0$, mean \tpro rises from $88.2\%$ to $97.4\%$. The largest improvement occurs under \llmrewrite, where \tpro jumps from $66.4\%$ to $91.8\%$ ($+25.4$pts). Mean \tprf rises from $96.6\%$ to $99.5\%$.

This additional detection robustness comes with substantially higher perplexity. Corpus PPL increases from $6.90$ to $34.66$, and \pplratio from $0.715$ to $3.591$, which is approximately a $\times$ gain in both cases. Hence, we conclude that the higher temperature produces a more diverse candidate set, allowing new sentences to better align with the field orientation and build a stronger watermark score. However, since there is no other restriction on them, the selected candidates can be highly unnatural and awkwardly phrased, reducing the overall quality.
\begin{table*}[h]
\centering
\scriptsize
\setlength{\tabcolsep}{2.5pt}
\renewcommand{\arraystretch}{1.08}
\caption{Effect of sampling temperature on \semfieldpl. Best values are bolded.}
\label{tab:semfield-temperature-ablation-no-nonce-roc500}
\resizebox{\textwidth}{!}{%
\begin{tabular}{lccccc|cc|c}
\toprule
Temp.
& \clean
& \pegasus
& \parrot
& \dipper
& \llmrewrite
& \shortstack{Mean\\\tpro}
& \shortstack{Mean\\\tprf}
& \shortstack{Corpus PPL $\downarrow$\\(\pplratio)} \\
\midrule
$0.7$
& \best{100.0} / 99.0 / \best{100.0}
& 99.7 / 94.2 / 99.0
& 99.8 / 95.8 / 99.2
& 99.1 / 85.4 / 95.8
& 97.8 / 66.4 / 89.2
& 88.2
& 96.6
& \best{6.90 (0.715)} \\
$1.0$
& \best{100.0} / \best{100.0} / \best{100.0}
& \best{100.0} / \best{99.8} / \best{100.0}
& \best{100.0} / \best{99.6} / \best{100.0}
& \best{99.7} / \best{95.8} / \best{99.0}
& \best{99.7} / \best{91.8} / \best{98.6}
& \best{97.4}
& \best{99.5}
& 34.66 (3.591) \\
\bottomrule
\end{tabular}%
}
\end{table*}

\paragraph{Increasing candidates.}
\begin{wrapfigure}{r}{0.48\textwidth}
\centering
\includegraphics[width=\linewidth]{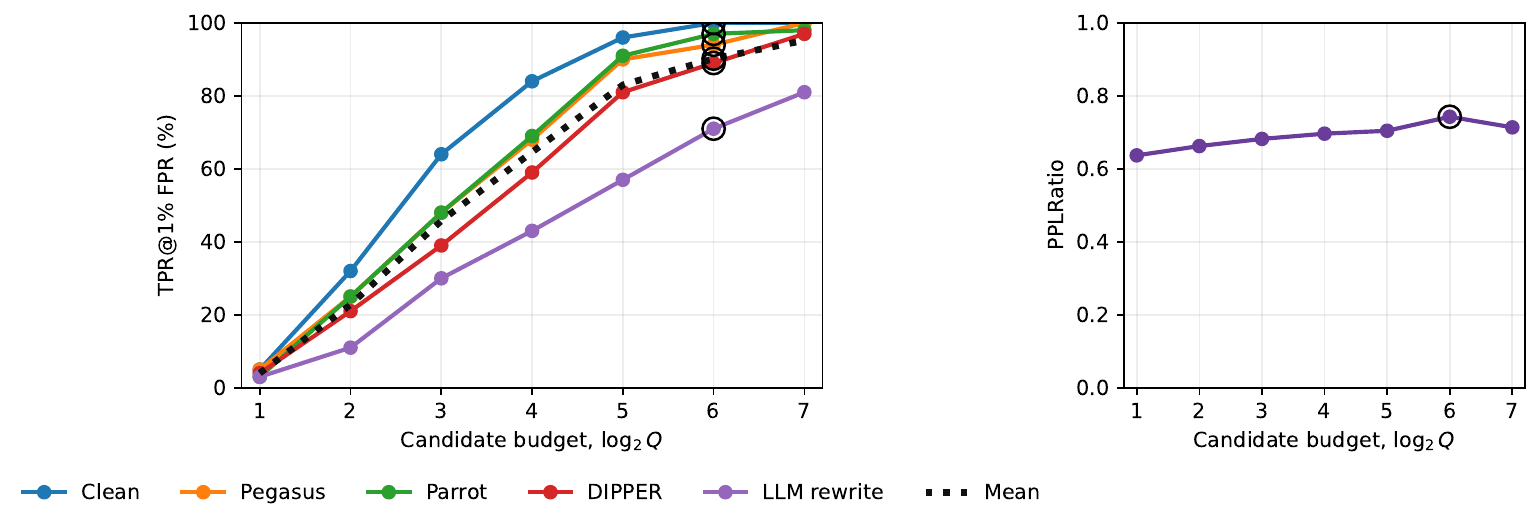}
\caption{\tpro vs. $Q$.}
\label{fig:q-sweep}
\end{wrapfigure}
At each generation step, both algorithms select from $Q$ candidate sentences using the resulting document score. Increasing $Q$ gives it more opportunities to find a candidate that reinforces the keyed direction, while requiring more candidate generations and evaluations. \Cref{fig:q-sweep} shows the effect on 100 documents with \llama as $Q$ doubles from $2$ to $128$ under the standard generation configuration.

Clean detection rises rapidly at smaller budgets and begins to level off at larger ones ($\geq 32$). Mean \tpro increases from $4.0\%$ at $Q=2$ to $90.2\%$ at $Q=64$, then to $95.2\%$ at $Q=128$. At $Q=64$, \clean reaches $100\%$, while \pegasus, \parrot, \dipper, and \llmrewrite obtain $94\%$, $97\%$, $89\%$, and $71\%$, respectively. Interestingly, \llmrewrite shows an almost linear trend against the $\log_2 Q$, with \tpro reaching $80.2\%$ at $Q=128$. Investigating this further, the best-fit line displays a slope of $13.61$ with an $R^2$ of $0.994$. This indicates significant diminishing returns with $Q$, as based on the current trend, the best-fit extrapolation predicts a $Q \geq 182$ for \tpro$\geq 90\%$.

Across the sweep, \pplratio stays between $0.637$ and $0.743$, while displaying a minor upward trend from $Q=2$ to $Q=64$, with a slight dip back to $0.714$ at $Q=128$. Overall, we argue that it shows no significant monotonic improvement with the candidate budget, unlike the \tpr scores. Therefore, there is no impact on the quality of generations at higher $Q$, but we are mainly constrained by the computation costs.

\subsection{Practical Deployment Scenarios}
\paragraph{Detection with Independently/Empirically Calibrated Thresholds.} 
In a \textit{practical} deployment, a detector must identify watermarked text using a threshold established before the document is received, even when its edit history is unknown. Hence, we complement the ROC-based results presented above with an evaluation that \textit{estimates} and \textit{freezes} the threshold of each detector \textit{a priori} across clean and attacked texts. To estimate these thresholds, we select 10,000 independent null documents from \textsc{C4}, along with 2,000 held-out documents for validation. For each model--method pair, we estimate the cut-off score by running its native detector on the 10,000 calibration documents and sorting the resulting scores from highest to lowest. For the nominal $1\%$ \fpr, we place the cutoff immediately above the 101$^{st}$ highest score, allowing at most 100 calibration documents to pass the threshold. Similarly, the nominal $5\%$ cutoff is placed immediately above the 501$^{st}$ highest score, allowing at most 500 detections. We then freeze these cutoffs and use the 2,000 held-out documents to measure the realized \fpr on unseen null text. Post-verification, the same frozen cutoffs are subsequently applied to the watermarked documents in every clean and attacked condition.

\Cref{tab:llama32-1b-appendix-null10k-current,tab:qwen3-1.7b-appendix-null10k-current,tab:opt-1.3b-appendix-null10k-current} demonstrate the aforementioned results, and we find that both \semfield variants once again occupy the top two positions in mean \tpro and \tprf across all three models. At the nominal $1\%$ \fpr, \semfieldpl achieves mean detection rates of $89.3\%$, $91.2\%$, and $89.5\%$. The strongest baseline, \swordstamp, obtains $72.5\%$, $73.4\%$, and $76.8\%$. Regular \semfield also remains ahead of every baseline, with means of $87.1\%$, $87.8\%$, and $87.1\%$. Thus, the overall detection advantage observed in the main evaluation also holds when thresholds are frozen using a substantially larger independent null corpus. The advantage is well established under full-document rewriting as well. Against \llmrewrite, \semfieldpl obtains $68.6\%$, $78.0\%$, and $71.8\%$ \tpro, compared with $43.6\%$, $47.0\%$, and $46.6\%$ for the strongest baseline, \samark. PL also improves mean \tpro over regular \semfield by $2.2$, $3.4$, and $2.4\%$, and achieves the highest mean \tprf of $96.8\%$, $96.7\%$, and $96.6\%$. Unlike the previous evaluation, PL has a clear advantage over the \semfield on \llmrewrite with respect to \tpro with a gain of $2.4$, $9.0$, and $4.4\%$, indicating its consistent advantage in a practical use case.

\paragraph{Cross-Key Detection.} 
Lastly, consider another scenario where $N$ providers employ \semfieldpl, with their own keys, $K_0, K_1 \dots K_{N-1}$. Following this setup, we now analyze how often a document generated with $K_i$ gets falsely detected by another key $K_j$. In an \textit{ideal} scenario, every document should be detected by its own key, and not by another. We simulate this with $100$ unique keys, and generate the 500 documents with \llama for each key. Next, we evaluate every generated document under all 100 keys, giving $5$ million document--key evaluations. Following from the previous experiment, each key uses its own detection threshold calibrated to a nominal $1\%$ \fpr on the same $10{,}000$ independent null documents. 

Averaging the diagonal entries of the resulting detection/confusion matrix ($K_{\text{gen}} = K_{\text{det}}$), the generation key detects $99.518\%$ of its documents. Similarly, averaging the off-diagonal entries ($K_{\text{gen}} \neq K_{\text{det}}$), an unrelated key detects these documents in $1.065\%$ of cases, which is close to the nominal false-positive operating point. These results show high detection under the generation key and cross-key detection near the expected false-alarm rate.

\section{Conclusion}
In this work, we proposed \semfield, a simple, linear yet robust semantic watermarking strategy for LLMs. \semfield uses a continuous Gaussian vector as a carrier, and iteratively selects new sentences that best improve the normalized aggregate document projection onto it. We also extend it to a two-sided variant, \semfieldpl, which first determines the orientation of the vector to align to using the first candidate, and then iteratively selects new sentences that best improve the aggregate document projection only along that direction. With an extensive comparison across 12 baselines and 4 paraphrasing attacks, we show the efficacy of our methods across 3 models, in which both variants occupy the top two places in terms of mean \tpro and MAUVE scores, while displaying a perplexity score close to the median of all the baselines. Further, we also discussed the robustness of the variants to various structural changes that do not involve any paraphrasing, and demonstrate an exact invariance to reordering and provable score stability under small perturbations to the content. Finally, we hope that our work inspires more studies on robust watermarking strategies for LLMs, which are well detectable under rewriting, and also maintain sufficient quality during generation.
\section{Limitations and Scope of Work}
Our work is subject to the following main limitations. We openly acknowledge that the watermarking experiments were targeted towards vanilla language modeling and base models, but this was done to align with the provided baselines and their evaluations. Similarly, following previous works, \semfield is also built and tested only for English, and it uses the sentence-based structure. We acknowledge that this might not be feasible for other languages where a clear sentence structure/boundary is not available (Ex, Chinese). Additionally, we only test our framework across three models. We hold the singular embedding model and dataset constant across all runs, as most baselines have reported it in their implementations, but all our experiments are agnostic to it. Due to computational constraints, we select a 500-document cohort, which, however translated to $4\times$ the number of evaluations accounting for the attacks. Future work in this direction can include more efficient sampling strategies and multi-dimensional candidate selection to account for safety, helpfulness, correctness, and low-entropy generations, such as math or code for RLVR and thinking models. Finally, we would like to highlight that we \textbf{do not} claim this work as an overall SoTA across all methods, models, and use cases, but present it as a \textbf{proof-of-concept} for a new and simple technique in line with previous studies.

\bibliography{iclr2027_conference}

@InProceedings{pmlr-v202-kirchenbauer23a,
  title = 	 {A Watermark for Large Language Models},
  author =       {Kirchenbauer, John and Geiping, Jonas and Wen, Yuxin and Katz, Jonathan and Miers, Ian and Goldstein, Tom},
  booktitle = 	 {Proceedings of the 40th International Conference on Machine Learning},
  pages = 	 {17061--17084},
  year = 	 {2023},
  editor = 	 {Krause, Andreas and Brunskill, Emma and Cho, Kyunghyun and Engelhardt, Barbara and Sabato, Sivan and Scarlett, Jonathan},
  volume = 	 {202},
  series = 	 {Proceedings of Machine Learning Research},
  month = 	 {23--29 Jul},
  publisher =    {PMLR},
  url = 	 {https://proceedings.mlr.press/v202/kirchenbauer23a.html}
}

@article{
kuditipudi2024robust,
title={Robust Distortion-free Watermarks for Language Models},
author={Rohith Kuditipudi and John Thickstun and Tatsunori Hashimoto and Percy Liang},
journal={Transactions on Machine Learning Research},
issn={2835-8856},
year={2024},
url={https://openreview.net/forum?id=FpaCL1MO2C},
note={}
}

@article{Dathathri2024,
    author={Dathathri, Sumanth and See, Abigail and Ghaisas, Sumedh and Huang, Po-Sen and McAdam, Rob and Welbl, Johannes and Bachani, Vandana and Kaskasoli, Alex and Stanforth, Robert and Matejovicova, Tatiana and Hayes, Jamie and Vyas, Nidhi and Merey, Majd Al and Brown-Cohen, Jonah and Bunel, Rudy and Balle, Borja and Cemgil, Taylan and Ahmed, Zahra and Stacpoole, Kitty and Shumailov, Ilia and Baetu, Ciprian and Gowal, Sven and Hassabis, Demis and Kohli, Pushmeet},
    title={Scalable watermarking for identifying large language model outputs},
    journal={Nature},
    year={2024},
    month={Oct},
    day={01},
    volume={634},
    number={8035},
    pages={818-823},
    issn={1476-4687},
    doi={10.1038/s41586-024-08025-4},
    url={https://doi.org/10.1038/s41586-024-08025-4}
}

@inproceedings{
liu2024a,
title={A Semantic Invariant Robust Watermark for Large Language Models},
author={Aiwei Liu and Leyi Pan and Xuming Hu and Shiao Meng and Lijie Wen},
booktitle={The Twelfth International Conference on Learning Representations},
year={2024},
url={https://openreview.net/forum?id=6p8lpe4MNf}
}

@inproceedings{hou-etal-2024-k,
    title = "k-{S}em{S}tamp: A Clustering-Based Semantic Watermark for Detection of Machine-Generated Text",
    author = "Hou, Abe  and
      Zhang, Jingyu  and
      Wang, Yichen  and
      Khashabi, Daniel  and
      He, Tianxing",
    editor = "Ku, Lun-Wei  and
      Martins, Andre  and
      Srikumar, Vivek",
    booktitle = "Findings of the Association for Computational Linguistics: ACL 2024",
    month = aug,
    year = "2024",
    address = "Bangkok, Thailand",
    publisher = "Association for Computational Linguistics",
    url = "https://aclanthology.org/2024.findings-acl.98/",
    doi = "10.18653/v1/2024.findings-acl.98",
    pages = "1706--1715"
}

@inproceedings{hou-etal-2024-semstamp,
    title = "{S}em{S}tamp: A Semantic Watermark with Paraphrastic Robustness for Text Generation",
    author = "Hou, Abe  and
      Zhang, Jingyu  and
      He, Tianxing  and
      Wang, Yichen  and
      Chuang, Yung-Sung  and
      Wang, Hongwei  and
      Shen, Lingfeng  and
      Van Durme, Benjamin  and
      Khashabi, Daniel  and
      Tsvetkov, Yulia",
    editor = "Duh, Kevin  and
      Gomez, Helena  and
      Bethard, Steven",
    booktitle = "Proceedings of the 2024 Conference of the North American Chapter of the Association for Computational Linguistics: Human Language Technologies (Volume 1: Long Papers)",
    month = jun,
    year = "2024",
    address = "Mexico City, Mexico",
    publisher = "Association for Computational Linguistics",
    url = "https://aclanthology.org/2024.naacl-long.226/",
    doi = "10.18653/v1/2024.naacl-long.226",
    pages = "4067--4082"
}

@inproceedings{dabiriaghdam-wang-2025-simmark,
    title = "{S}im{M}ark: A Robust Sentence-Level Similarity-Based Watermarking Algorithm for Large Language Models",
    author = "Dabiriaghdam, Amirhossein  and
      Wang, Lele",
    editor = "Christodoulopoulos, Christos  and
      Chakraborty, Tanmoy  and
      Rose, Carolyn  and
      Peng, Violet",
    booktitle = "Proceedings of the 2025 Conference on Empirical Methods in Natural Language Processing",
    month = nov,
    year = "2025",
    address = "Suzhou, China",
    publisher = "Association for Computational Linguistics",
    url = "https://aclanthology.org/2025.emnlp-main.1567/",
    doi = "10.18653/v1/2025.emnlp-main.1567",
    pages = "30785--30806",
    ISBN = "979-8-89176-332-6"
}

@inproceedings{
huo2026pmark,
title={{PM}ark: Towards Robust and Distortion-free Semantic-level Watermarking with Channel Constraints},
author={Jiahao Huo and Shuliang Liu and Bin Wang and Junyan Zhang and Yibo Yan and Aiwei Liu and Xuming Hu and Mingxun Zhou},
booktitle={The Fourteenth International Conference on Learning Representations},
year={2026},
url={https://openreview.net/forum?id=EhDgP69DJG}
}

@inproceedings{
li2026alimark,
title={AliMark: Enhancing Robustness of Sentence-Level Watermarking Against Text Paraphrasing},
author={Yuexin Li and Wenjie Qu and Linyu Wu and Yulin Chen and Yufei He and Tri Cao and Bryan Hooi and Jiaheng Zhang},
booktitle={Forty-third International Conference on Machine Learning},
year={2026},
url={https://openreview.net/forum?id=jQmlwZSPuw}
}

@misc{huo2026samarkselfanchoredtextwatermarking,
      title={SAMark: A Self-Anchored Text Watermarking with Paragraph-Level Paraphrase Robustness}, 
      author={Jiahao Huo and Wenjie Qu and Yibo Yan and Kening Zheng and Jiaheng Zhang and Xuming Hu and Philip S. Yu and Mingxun Zhou},
      year={2026},
      eprint={2605.25796},
      archivePrefix={arXiv},
      primaryClass={cs.CR},
      url={https://arxiv.org/abs/2605.25796}, 
}

@misc{diaa2026semanticwatermarkingorderrobustdetection,
      title={Semantic Watermarking with Order-Robust Detection over Sub-sentence Units}, 
      author={Abdulrahman Diaa and Jonathan Petit and Florian Kerschbaum},
      year={2026},
      eprint={2608.27666},
      archivePrefix={arXiv},
      primaryClass={cs.CR},
      url={https://arxiv.org/abs/2608.27666}, 
}

@misc{xue2025proenablingpreciserobust,
      title={PRO: Enabling Precise and Robust Text Watermark for Open-Source LLMs}, 
      author={Jiaqi Xue and Yifei Zhao and Mansour Al Ghanim and Shangqian Gao and Ruimin Sun and Qian Lou and Mengxin Zheng},
      year={2025},
      eprint={2510.23891},
      archivePrefix={arXiv},
      primaryClass={cs.CR},
      url={https://arxiv.org/abs/2510.23891}, 
}

@inproceedings{3698900.3699002,
author = {Zhang, Ruisi and Hussain, Shehzeen Samarah and Neekhara, Paarth and Koushanfar, Farinaz},
title = {REMARK-LLM: a robust and efficient watermarking framework for generative large language models},
year = {2024},
isbn = {978-1-939133-44-1},
publisher = {USENIX Association},
address = {USA},
booktitle = {Proceedings of the 33rd USENIX Conference on Security Symposium},
articleno = {102},
numpages = {18},
location = {Philadelphia, PA, USA},
series = {SEC '24}
}

@inproceedings{
xu2025learning,
title={Learning to watermark {LLM}-Generated Text Via Reinforcement Learning},
author={Xiaojun Xu and Yuanshun Yao and Yang Liu},
booktitle={The 1st Workshop on GenAI Watermarking},
year={2025},
url={https://openreview.net/forum?id=dTlCMiEOdk}
}

@misc{prithivida2021parrot,
  author       = {Prithiviraj Damodaran},
  title        = {Parrot: Paraphrase generation for NLU.},
  year         = 2021,
  version      = {v1.0}
}

@misc{qwen38,
    title = {{Qwen3.8-Max}: A New Bar for Coding and Cowork},
    url = {https://qwen.ai/blog?id=qwen3.8},
    author = {{Qwen Team}},
    month = {August},
    year = {2026}
}

@misc{grattafiori2024llama3herdmodels,
      title={The Llama 3 Herd of Models}, 
      author={Aaron Grattafiori and Abhimanyu Dubey and Abhinav Jauhri and Abhinav Pandey and Abhishek Kadian and Ahmad Al-Dahle and Aiesha Letman and Akhil Mathur and Alan Schelten et al.},
      year={2024},
      eprint={2407.21783},
      archivePrefix={arXiv},
      primaryClass={cs.AI},
      url={https://arxiv.org/abs/2407.21783}, 
}

@misc{zhang2022opt,
      title={OPT: Open Pre-trained Transformer Language Models}, 
      author={Susan Zhang and Stephen Roller and Naman Goyal and Mikel Artetxe and Moya Chen and Shuohui Chen and Christopher Dewan and Mona Diab and Xian Li and Xi Victoria Lin and Todor Mihaylov and Myle Ott and Sam Shleifer and Kurt Shuster and Daniel Simig and Punit Singh Koura and Anjali Sridhar and Tianlu Wang and Luke Zettlemoyer},
      year={2022},
      eprint={2205.01068},
      archivePrefix={arXiv},
      primaryClass={cs.CL}
}

@misc{aaronson2023watermarking,
  author       = {Scott Aaronson},
  title        = {Watermarking of Large Language Models},
  year         = {2023},
  howpublished = {Talk at the Simons Institute for the Theory of Computing},
  url          = {https://simons.berkeley.edu/talks/scott-aaronson-ut-austin-openai-2023-08-17},
  note         = {Slides available at \url{https://www.scottaaronson.com/talks/watermark.ppt}}
}

@inproceedings{chang-etal-2024-postmark,
    title = "{P}ost{M}ark: A Robust Blackbox Watermark for Large Language Models",
    author = "Chang, Yapei  and
      Krishna, Kalpesh  and
      Houmansadr, Amir  and
      Wieting, John Frederick  and
      Iyyer, Mohit",
    editor = "Al-Onaizan, Yaser  and
      Bansal, Mohit  and
      Chen, Yun-Nung",
    booktitle = "Proceedings of the 2024 Conference on Empirical Methods in Natural Language Processing",
    month = nov,
    year = "2024",
    address = "Miami, Florida, USA",
    publisher = "Association for Computational Linguistics",
    url = "https://aclanthology.org/2024.emnlp-main.506/",
    doi = "10.18653/v1/2024.emnlp-main.506",
    pages = "8969--8987"
}

@inproceedings{
krishna2023paraphrasing,
title={Paraphrasing evades detectors of {AI}-generated text, but retrieval is an effective defense},
author={Kalpesh Krishna and Yixiao Song and Marzena Karpinska and John Frederick Wieting and Mohit Iyyer},
booktitle={Thirty-seventh Conference on Neural Information Processing Systems},
year={2023},
url={https://openreview.net/forum?id=WbFhFvjjKj}
}

@misc{pham2025tensorsketchfastscalable,
      title={Tensor Sketch: Fast and Scalable Polynomial Kernel Approximation}, 
      author={Ninh Pham and Rasmus Pagh},
      year={2025},
      eprint={2505.08146},
      archivePrefix={arXiv},
      primaryClass={cs.DS},
      url={https://arxiv.org/abs/2505.08146}, 
}

@misc{yang2025qwen3technicalreport,
      title={Qwen3 Technical Report}, 
      author={An Yang and Anfeng Li and Baosong Yang and Beichen Zhang and Binyuan Hui and Bo Zheng and Bowen Yu and Chang Gao and Chengen Huang and Chenxu Lv and Chujie Zheng and Dayiheng Liu and Fan Zhou and Fei Huang and Feng Hu and Hao Ge and Haoran Wei and Huan Lin and Jialong Tang and Jian Yang and Jianhong Tu and Jianwei Zhang and Jianxin Yang and Jiaxi Yang and Jing Zhou and Jingren Zhou and Junyang Lin and Kai Dang and Keqin Bao and Kexin Yang and Le Yu and Lianghao Deng and Mei Li and Mingfeng Xue and Mingze Li and Pei Zhang and Peng Wang and Qin Zhu and Rui Men and Ruize Gao and Shixuan Liu and Shuang Luo and Tianhao Li and Tianyi Tang and Wenbiao Yin and Xingzhang Ren and Xinyu Wang and Xinyu Zhang and Xuancheng Ren and Yang Fan and Yang Su and Yichang Zhang and Yinger Zhang and Yu Wan and Yuqiong Liu and Zekun Wang and Zeyu Cui and Zhenru Zhang and Zhipeng Zhou and Zihan Qiu},
      year={2025},
      eprint={2505.09388},
      archivePrefix={arXiv},
      primaryClass={cs.CL},
      url={https://arxiv.org/abs/2505.09388}, 
}

@inproceedings{NEURIPS2021_260c2432,
 author = {Pillutla, Krishna and Swayamdipta, Swabha and Zellers, Rowan and Thickstun, John and Welleck, Sean and Choi, Yejin and Harchaoui, Zaid},
 booktitle = {Advances in Neural Information Processing Systems},
 editor = {M. Ranzato and A. Beygelzimer and Y. Dauphin and P.S. Liang and J. Wortman Vaughan},
 pages = {4816--4828},
 publisher = {Curran Associates, Inc.},
 title = {MAUVE: Measuring the Gap Between Neural Text and Human Text using Divergence Frontiers},
 url = {https://proceedings.neurips.cc/paper_files/paper/2021/file/260c2432a0eecc28ce03c10dadc078a4-Paper.pdf},
 volume = {34},
 year = {2021}
}

@InProceedings{zhang2020pegasus,
  title = 	 {{PEGASUS}: Pre-training with Extracted Gap-sentences for Abstractive Summarization},
  author =       {Zhang, Jingqing and Zhao, Yao and Saleh, Mohammad and Liu, Peter},
  booktitle = 	 {Proceedings of the 37th International Conference on Machine Learning},
  pages = 	 {11328--11339},
  year = 	 {2020},
  editor = 	 {III, Hal Daumé and Singh, Aarti},
  volume = 	 {119},
  series = 	 {Proceedings of Machine Learning Research},
  month = 	 {13--18 Jul},
  publisher =    {PMLR},
  url = 	 {https://proceedings.mlr.press/v119/zhang20ae.html}
}

@inproceedings{hanley2024machinemade,
  title={Machine-Made Media: Monitoring the Mobilization of Machine-Generated Articles on Misinformation and Mainstream News Websites},
  author={Hanley, Hans W. A. and Durumeric, Zakir},
  booktitle={Proceedings of the International AAAI Conference on Web and Social Media},
  volume={18},
  pages={542--556},
  year={2024},
  doi={10.1609/icwsm.v18i1.31333},
  url={https://ojs.aaai.org/index.php/ICWSM/article/view/31333}
}

@misc{dolezal2026impact,
  title={The Impact of AI-Generated Text on the Internet},
  author={Dolezal, Jonas and Alam, Sawood and Graham, Mark and Bohacek, Maty},
  year={2026},
  eprint={2604.26965},
  archivePrefix={arXiv},
  primaryClass={cs.CY},
  url={https://arxiv.org/abs/2604.26965}
}

@misc{anthropic2026claudewatermark,
  author       = {{Anthropic}},
  title        = {How {Claude's} text watermark works},
  howpublished = {\url{https://www.anthropic.com/news/claude-text-watermark}},
  year         = {2026},
  month        = aug,
  day          = {14},
  note         = {Accessed: 2026-09-29}
}

@misc{ec2026codeofpractice,
  author       = {{European Commission}},
  title        = {Code of {Practice} on {Transparency} of {AI-generated} {Content}},
  howpublished = {\url{https://digital-strategy.ec.europa.eu/en/policies/code-practice-ai-generated-content}},
  year         = {2026},
  month        = jul,
  day          = {31},
  note         = {Accessed: 2026-09-29}
}

@inproceedings{ICLR2025_ea94957d,
 author = {Si, Chenglei and Yang, Diyi and Hashimoto, Tatsunori},
 booktitle = {International Conference on Learning Representations},
 editor = {Y. Yue and A. Garg and N. Peng and F. Sha and R. Yu},
 pages = {94003--94092},
 title = {Can LLMs Generate Novel Research Ideas? A Large-Scale Human Study with 100+ NLP Researchers},
 url = {https://proceedings.iclr.cc/paper_files/paper/2025/file/ea94957d81b1c1caf87ef5319fa6b467-Paper-Conference.pdf},
 volume = {2025},
 year = {2025}
}

@article{10.1145/3771937,
author = {Kashif, Syed Mohammad and Liang, Peng and Tahir, Amjed},
title = {On Developers’ Self-Declaration of AI-Generated Code: An Analysis of Practices},
year = {2026},
issue_date = {July 2026},
publisher = {Association for Computing Machinery},
address = {New York, NY, USA},
volume = {35},
number = {7},
issn = {1049-331X},
url = {https://doi.org/10.1145/3771937},
doi = {10.1145/3771937},
journal = {ACM Trans. Softw. Eng. Methodol.},
month = jul,
articleno = {200},
numpages = {37}
}

@article{Liang2025,
  author   = {Liang, Weixin and Zhang, Yaohui and Wu, Zhengxuan and Lepp, Haley and Ji, Wenlong and Zhao, Xuandong and Cao, Hancheng and Liu, Sheng and He, Siyu and Huang, Zhi and Yang, Diyi and Potts, Christopher and Manning, Christopher D. and Zou, James},
  title    = {Quantifying large language model usage in scientific papers},
  journal  = {Nature Human Behaviour},
  year     = {2025},
  month    = dec,
  volume   = {9},
  number   = {12},
  pages    = {2599--2609},
  doi      = {10.1038/s41562-025-02273-8},
  url      = {https://doi.org/10.1038/s41562-025-02273-8},
  issn     = {2397-3374}
}

@article{Stutz2026,
  author  = {Stutz, David and Cowen-Rivers, Alexander I. and Ortiz-Jimenez, Guillermo
             and Ratcliff, Jeremy and Zambaldi, Vinicius and Willmore, Lindsay
             and Abramson, Josh and Patani, Harshnira and Kouridi, Christina
             and Stimberg, Florian and Vecerik, Mel and Chu, Alex
             and Singh, Sukhdeep and Dathathri, Sumanth and Papa, Eliseo
             and De Bortoli, Valentin and Doucet, Arnaud and Hassabis, Demis
             and Wang, Jue and Gowal, Sven and Kohli, Pushmeet},
  title   = {Function-preserving watermarking of AI-generated proteins},
  journal = {Nature},
  year    = {2026},
  month   = sep,
  doi     = {10.1038/s41586-026-10965-y},
  url     = {https://doi.org/10.1038/s41586-026-10965-y},
  issn    = {1476-4687}
}

@article{
doi:10.1126/science.adz9311,
author = {Simone Daniotti  and Johannes Wachs  and Xiangnan Feng  and Frank Neffke },
title = {Who is using AI to code? Global diffusion and impact of generative AI},
journal = {Science},
volume = {391},
number = {6787},
pages = {831-835},
year = {2026},
doi = {10.1126/science.adz9311},
URL = {https://www.science.org/doi/abs/10.1126/science.adz9311},
eprint = {https://www.science.org/doi/pdf/10.1126/science.adz9311}}
\bibliographystyle{iclr2027_conference}

\appendix
\section{Propositions, Theorems and Proofs}
\label{app:prop-thm-proofs}

Throughout this section, \(K\) denotes the secret key and
\(\omega_K\in\mathbb{R}^d\) denotes the corresponding keyed direction. All structural results below instead hold deterministically for
any fixed \(\omega_K\).

\begin{theorem}[Gaussian field marginal]
\label{thm:gaussian-field-marginal}
Let \(x\in\mathbb{R}^d\) be a fixed unit vector and suppose that
\[
\omega_K\sim\mathcal{N}(0,I_d).
\]
Define the keyed semantic field
\[
f_K(x)=\omega_K^{\mathsf T}x.
\]
Then
\[
f_K(x)\sim\mathcal{N}(0,1).
\]
\end{theorem}

\begin{proof}
Since \(f_K(x)\) is a linear combination of independent Gaussian
variables, it is Gaussian. Its mean is
\[
\begin{aligned}
\mathbb{E}_K[f_K(x)]
&=
\mathbb{E}_K[\omega_K^{\mathsf T}x] \\
&=
\mathbb{E}_K[\omega_K]^{\mathsf T}x \\
&=0.
\end{aligned}
\]

Its variance is
\[
\begin{aligned}
\operatorname{Var}_K(f_K(x))
&=
\mathbb{E}_K[f_K(x)^2]
-
\left(\mathbb{E}_K[f_K(x)]\right)^2 \\
&=
\mathbb{E}_K\!\left[
\left(\omega_K^{\mathsf T}x\right)^2
\right] \\
&=
x^{\mathsf T}
\mathbb{E}_K[\omega_K\omega_K^{\mathsf T}]
x \\
&=
x^{\mathsf T}I_dx \\
&=
\lVert x\rVert_2^2 \\
&=1.
\end{aligned}
\]

Therefore,
\[
f_K(x)\sim\mathcal{N}(0,1).
\]
\end{proof}

\begin{theorem}[Covariance of keyed semantic-field scores]
\label{thm:covariance-of-keyed-semantic-field-scores}
Let \(x,y\in\mathbb{R}^d\) be fixed embeddings and suppose that
\[
\omega_K\sim\mathcal{N}(0,I_d).
\]
For \(z\in\mathbb{R}^d\), define
\[
f_K(z)=\omega_K^{\mathsf T}z.
\]
Then
\[
\operatorname{Cov}_K\!\left(f_K(x),f_K(y)\right)
=
x^{\mathsf T}y.
\]
\end{theorem}

\begin{proof}
Because \(x\) and \(y\) are fixed and \(\mathbb{E}_K[\omega_K]=0\), both field scores have zero mean:
\[
\mathbb{E}_K[f_K(x)]=0,
\qquad
\mathbb{E}_K[f_K(y)]=0.
\]

Therefore,
\[
\begin{aligned}
\operatorname{Cov}_K\!\left(f_K(x),f_K(y)\right)
&=
\mathbb{E}_K\!\left[
f_K(x)f_K(y)
\right] \\
&=
\mathbb{E}_K\!\left[
\left(\omega_K^{\mathsf T}x\right)
\left(\omega_K^{\mathsf T}y\right)
\right] \\
&=
\mathbb{E}_K\!\left[
x^{\mathsf T}\omega_K\omega_K^{\mathsf T}y
\right] \\
&=
x^{\mathsf T}
\mathbb{E}_K[\omega_K\omega_K^{\mathsf T}]
y \\
&=
x^{\mathsf T}I_dy \\
&=
x^{\mathsf T}y.
\end{aligned}
\]
\end{proof}

\begin{theorem}[Conditional Gaussian null]
\label{thm:conditional-gaussian-null}
Let \(D\) be a null document with normalized sentence embeddings
\(x_1,\ldots,x_T\), and let
\[
A_T=\sum_{t=1}^{T}x_t.
\]
Assume that \(A_T\neq 0\), that \(D\) is independent of the randomly
drawn tested key \(K\), and that, 
\[
\omega_K\sim\mathcal{N}(0,I_d).
\]
Define
\[
u_T=\frac{A_T}{\lVert A_T\rVert_2}
\]
and
\[
Z_K=\omega_K^{\mathsf T}u_T.
\]
Then, under the null hypothesis,
\[
Z_K\mid D\sim\mathcal{N}(0,1).
\]
\end{theorem}

\begin{proof}
Conditional on the observed document \(D\), the aggregate \(A_T\) and the normalized vector \(u_T\) are fixed. Since \(D\) is independent of the tested key, conditioning on \(D\) does not change the distribution of \(\omega_K\).

The statistic is therefore a fixed linear combination of the Gaussian
vector \(\omega_K\):
\[
Z_K=\omega_K^{\mathsf T}u_T.
\]
Hence \(Z_K\mid D\) is Gaussian.

Its conditional mean is
\[
\begin{aligned}
\mathbb{E}_K[Z_K\mid D]
&=
\mathbb{E}_K[\omega_K^{\mathsf T}u_T\mid D] \\
&=
\mathbb{E}_K[\omega_K\mid D]^{\mathsf T}u_T \\
&=0.
\end{aligned}
\]

Its conditional variance is
\[
\begin{aligned}
\operatorname{Var}_K(Z_K\mid D)
&=
u_T^{\mathsf T}
\operatorname{Cov}_K(\omega_K\mid D)
u_T \\
&=
u_T^{\mathsf T}I_du_T \\
&=
\lVert u_T\rVert_2^2 \\
&=1.
\end{aligned}
\]

Therefore,
\[
Z_K\mid D\sim\mathcal{N}(0,1).
\]
\end{proof}

\begin{proposition}[Continuity of the keyed field]
\label{prop:field-continuity}
Let $f_K(x)=\omega_K^{\mathsf T}x$. For any
$u,v\in\mathbb{R}^d$,
\begin{equation}
|f_K(u)-f_K(v)|
\leq
\lVert\omega_K\rVert_2\,\lVert u-v\rVert_2.
\end{equation}
In particular, if $u$ and $v$ are unit vectors, then
\begin{equation}
|f_K(u)-f_K(v)|
\leq
\lVert\omega_K\rVert_2
\sqrt{2(1-u^{\mathsf T}v)}.
\end{equation}
Thus, $f_K$ is continuous for every fixed key.
\end{proposition}

\begin{proof}
By linearity and the Cauchy--Schwarz inequality,
\begin{equation}
|f_K(u)-f_K(v)|
=
|\omega_K^{\mathsf T}(u-v)|
\leq
\lVert\omega_K\rVert_2\,\lVert u-v\rVert_2.
\end{equation}
For unit vectors,
$\lVert u-v\rVert_2^2=2(1-u^{\mathsf T}v)$, giving the second bound.

To verify continuity directly, given any $\varepsilon>0$, choose $\delta=\frac{\varepsilon}{(\lVert\omega_K\rVert_2+1)}$. Whenever $\lVert u-v\rVert_2<\delta$, the first bound gives $|f_K(u)-f_K(v)|<\varepsilon$.
\end{proof}

\begin{proposition}[Exact invariance to sentence reordering]
\label{prop:reordering-invariance}
Let a document contain \(T\) intact sentences with fixed embeddings \(x_1,\ldots,x_T\), and let \(\pi\) be any permutation of \(\{1,\ldots,T\}\). Define
\[
A_T=\sum_{t=1}^{T}x_t,
\qquad
A_T^\pi=\sum_{t=1}^{T}x_{\pi(t)}.
\]
Assume \(A_T\neq 0\). For a fixed key \(K\), define, for any nonzero
aggregate \(A\),
\[
Z_K(A)=\omega_K^{\mathsf T}\frac{A}{\lVert A\rVert_2},
\qquad
H_K(A)=|Z_K(A)|.
\]
Then
\[
Z_K(A_T^\pi)=Z_K(A_T),
\qquad
H_K(A_T^\pi)=H_K(A_T).
\]
Consequently, detection by either \semfield, using \(Z_K\), or \semfieldpl, using \(H_K\), is exactly invariant to reordering intact sentences.
\end{proposition}

\begin{proof}
Because \(\pi\) is a permutation, it only changes the order of the terms in the finite sum. Vector addition is commutative, so
\[
A_T^\pi
=\sum_{t=1}^{T}x_{\pi(t)}
=\sum_{t=1}^{T}x_t
=A_T.
\]
Since \(A_T\neq 0\), the normalized aggregates are also equal. Hence
\[
Z_K(A_T^\pi)
=\omega_K^{\mathsf T}
\frac{A_T^\pi}{\lVert A_T^\pi\rVert_2}
=\omega_K^{\mathsf T}
\frac{A_T}{\lVert A_T\rVert_2}
=Z_K(A_T).
\]
Taking absolute values gives \(H_K(A_T^\pi)=H_K(A_T)\).
\end{proof}

\begin{proposition}[Deterministic stability under aggregate perturbations]
\label{prop:aggregate-stability}
Let \(A\neq 0\) be the semantic aggregate of a document, and let
\[
A'=A+\Delta\neq 0
\]
be the aggregate after an insertion or deletion attack. Define
\[
u=\frac{A}{\lVert A\rVert_2},
\qquad
u'=\frac{A'}{\lVert A'\rVert_2}.
\]
For a fixed key \(K\), define the signed scores and their absolute values:
\[
Z_K(A)=\omega_K^{\mathsf T}u,
\qquad
Z_K(A')=\omega_K^{\mathsf T}u',
\qquad
H_K(A)=|Z_K(A)|,
\qquad
H_K(A')=|Z_K(A')|.
\]

Then
\[
\lVert u'-u\rVert_2
\leq
\frac{2\lVert\Delta\rVert_2}
{\min\{\lVert A\rVert_2,\lVert A'\rVert_2\}}.
\]
Define
\[
\varepsilon_{\mathrm{edit}}
=
\frac{
2\lVert \omega_K\rVert_2\lVert\Delta\rVert_2
}{
\min\{\lVert A\rVert_2,\lVert A'\rVert_2\}
}.
\]
Then both variant scores obey
\[
|Z_K(A')-Z_K(A)|\leq\varepsilon_{\mathrm{edit}},
\qquad
|H_K(A')-H_K(A)|\leq\varepsilon_{\mathrm{edit}}.
\]
For either variant, let \(R_K\) be its score: \(R_K=Z_K\) for  \semfield and \(R_K=H_K\) for \semfieldpl. Let \(\tau\) be that variant's
detection threshold, with detection when \(R_K\geq\tau\). If \(R_K(A)-\tau>\varepsilon_{\mathrm{edit}}\), the document remains detected after the attack. Conversely, if \(\tau-R_K(A)>\varepsilon_{\mathrm{edit}}\), the attack cannot cause a detection.
\end{proposition}

\begin{proof}
For brevity, write
\[
a=\lVert A\rVert_2,
\qquad
a'=\lVert A'\rVert_2,
\qquad
\delta=\lVert\Delta\rVert_2.
\]
Since \(A'=A+\Delta\),
\[
A'-A=\Delta.
\]
Using \(u=A/a\) and \(u'=A'/a'\), we obtain
\[
\begin{aligned}
u'-u
&=
\frac{A'}{a'}-\frac{A}{a} \\
&=
\frac{A'-A}{a'}
+
A\left(\frac{1}{a'}-\frac{1}{a}\right).
\end{aligned}
\]

Taking Euclidean norms and applying the triangle inequality gives
\[
\begin{aligned}
\lVert u'-u\rVert_2
&\leq
\frac{\lVert A'-A\rVert_2}{a'}
+
\lVert A\rVert_2
\left|\frac{1}{a'}-\frac{1}{a}\right| \\
&=
\frac{\delta}{a'}
+
\frac{|a-a'|}{a'}.
\end{aligned}
\]

By the reverse triangle inequality,
\[
|a-a'|
=
\big|\lVert A\rVert_2-\lVert A'\rVert_2\big|
\leq
\lVert A-A'\rVert_2
=
\delta.
\]
Therefore,
\[
\lVert u'-u\rVert_2
\leq
\frac{2\delta}{a'}
\leq
\frac{2\delta}
{\min\{a,a'\}},
\]
which proves the normalized-direction bound.

Next, the Cauchy--Schwarz inequality gives
\[
\begin{aligned}
|Z_K(A')-Z_K(A)|
&=
|\omega_K^{\mathsf T}(u'-u)| \\
&\leq
\lVert \omega_K\rVert_2
\lVert u'-u\rVert_2 \\
&\leq\varepsilon_{\mathrm{edit}}.
\end{aligned}
\]
The reverse triangle inequality for real numbers then gives
\[
|H_K(A')-H_K(A)|
=\bigl||Z_K(A')|-|Z_K(A)|\bigr|
\leq |Z_K(A')-Z_K(A)|
\leq\varepsilon_{\mathrm{edit}}.
\]
For either score \(R_K\), these bounds imply
\[
R_K(A)-\varepsilon_{\mathrm{edit}}
\leq R_K(A')
\leq R_K(A)+\varepsilon_{\mathrm{edit}}.
\]
If \(R_K(A)-\tau>\varepsilon_{\mathrm{edit}}\), the lower bound yields \(R_K(A')>\tau\), so detection is preserved. If \(\tau-R_K(A)>\varepsilon_{\mathrm{edit}}\), the upper bound yields \(R_K(A')<\tau\), so the attack cannot cause a detection.
\end{proof}
\section{Higher-order polynomial fields: expressivity and cost}
\label{app:higher-order-fields}

The linear field used in this work is the \textit{lowest-degree nonconstant polynomial field} over the normalized embedding sphere. A natural extension is to use a higher-order polynomial field, which can represent interactions between embedding coordinates and may distinguish closely spaced candidate embeddings more strongly.

Let $x\in\mathbb{R}^{d}$ be a normalized sentence embedding. For an integer degree $m\geq 1$, we can define a higher-order feature representation as
\begin{equation}
    \psi_m(x)=x^{\otimes m},
\end{equation}
where $x^{\otimes m}$ denotes the $m$-fold tensor product of $x$. In the ideal Gaussian-key model, let $W_K^{(m)}$ be a secret random tensor whose entries are independent standard Gaussian variables. The corresponding sentence-level \textit{field score} is defined by
\begin{equation}
    f_K^{(m)}(x)
    =
    \left\langle
        W_K^{(m)},
        x^{\otimes m}
    \right\rangle_F,
\end{equation}
where $\langle\cdot,\cdot\rangle_F$ denotes the Frobenius inner product. The cases $m=1$, $m=2$, and $m=3$ correspond to linear, quadratic, and cubic fields, respectively.

For a document containing normalized sentence embeddings $x_1,\ldots,x_T$, the transformed document aggregate is
\begin{equation}
    A_T^{(m)}
    =
    \sum_{t=1}^{T}\psi_m(x_t).
\end{equation}
The document-level statistic is obtained by normalizing this aggregate and projecting it onto the keyed direction in the transformed feature space:
\begin{equation}
    Z_{K}^{(m)}
    =
    \frac{
        \left\langle W_{K}^{(m)},A_T^{(m)}\right\rangle_F
    }{
        \left\lVert A_T^{(m)}\right\rVert_F
    }.
\end{equation}
Thus, changing the polynomial degree changes the geometry used during both candidate selection and detection, while preserving the document-level construction of \semfield.

Because $x$ is normalized,
\begin{equation}
    \left\lVert x^{\otimes m}\right\rVert_F
    =
    \lVert x\rVert_2^m
    =
    1.
\end{equation}
Consequently, for every fixed $x$, the marginal field response over the ideal random-key draw satisfies
\begin{equation}
    f_K^{(m)}(x)\sim\mathcal{N}(0,1).
\end{equation}
For two fixed normalized embeddings $x$ and $y$, the covariance over the same ideal random-key draw between their field responses is
\begin{equation}
\begin{aligned}
    \operatorname{Cov}_K\!\left(f_K^{(m)}(x),f_K^{(m)}(y)\right)
    &=
    \left\langle
        x^{\otimes m},
        y^{\otimes m}
    \right\rangle_F \\
    &=
    \left(x^{\mathsf T}y\right)^m.
\end{aligned}
\end{equation}
The degree-$m$ field, therefore, induces the homogeneous polynomial kernel
\begin{equation}
    \kappa_m(x,y)
    =
    \left(x^{\mathsf T}y\right)^m.
\end{equation}

\paragraph{Expressivity.}
Let $c=x^{\mathsf T}y$ be the cosine similarity between two nearby embeddings. When $c$ is close to one,
\begin{equation}
\begin{aligned}
    c^m
    &=
    \left(1-(1-c)\right)^m \\
    &\approx
    1-m(1-c).
\end{aligned}
\end{equation}
The expected squared difference between their field responses is
\begin{equation}
\begin{aligned}
    \mathbb{E}_K\!\left[
        \left(f_K^{(m)}(x)-f_K^{(m)}(y)\right)^2
    \right]
    &=
    \operatorname{Var}_K\!\left(f_K^{(m)}(x)\right)
    +
    \operatorname{Var}_K\!\left(f_K^{(m)}(y)\right) \\
    &\quad
    -2\operatorname{Cov}_K\!\left(f_K^{(m)}(x),f_K^{(m)}(y)\right) \\
    &=
    2\left(1-c^m\right) \\
    &\approx
    2m(1-c).
\end{aligned}
\end{equation}
Hence, increasing the degree can enlarge the differences between the field responses of closely related candidates. A higher-order field may therefore provide a larger candidate-selection margin when the original embeddings occupy a narrow region of the embedding space.

But this additional sensitivity is not specific to generation candidates. If rewriting moves an embedding from $x$ to $y$, the same factor $1-c^m$ controls the change in their field responses. Increasing $m$ can therefore magnify both useful candidate differences and unwanted semantic movement introduced by rewriting. Higher-order fields consequently create an expressivity--stability trade-off rather than providing an unconditional increase in watermark robustness.

\paragraph{Computational cost.}
The ideal representation $x^{\otimes m}$ contains $d^m$ entries. Explicitly constructing and storing it requires $\mathcal{O}\!\left(d^m\right)$ space, while evaluating a dense degree-$m$ field requires $\mathcal{O}\!\left(d^m\right)$ operations per candidate. Even when tensor symmetry is exploited, the number of distinct degree-$m$ monomials is $\binom{d+m-1}{m}$, which remains large for modern sentence-embedding dimensions.

Approximate polynomial feature maps, such as \textsc{TensorSketch} \citep{pham2025tensorsketchfastscalable} or random Maclaurin features, can replace the explicit tensor by a fixed $p$-dimensional representation. This reduces storage for each feature representation to $\mathcal{O}(p)$ and makes higher-order fields computationally feasible. However, the resulting inner products only approximate $(x^{\mathsf T}y)^m$, introducing an additional approximation error whose size depends on $p$. The same public feature map must also be fixed and used consistently during generation and detection.

Higher-order fields, therefore, offer greater local expressivity at the cost of increased sensitivity, computation, storage, and, when approximated, feature-map error. The linear field studied in this work avoids these costs and retains direct, linear dependence on the embedding representation.
\section{Calibrated Results}
\label{app:calib-res}

\begin{table*}[h]
\centering
\scriptsize
\setlength{\tabcolsep}{2.0pt}
\renewcommand{\arraystretch}{1.10}
\caption{Fixed-threshold detection results for \llama. The same cutoff is used for clean and all rewriting conditions. Best values in a column are bolded, and the second best are underlined.}
\label{tab:llama32-1b-appendix-null10k-current}
\resizebox{\textwidth}{!}{%
\begin{tabular}{lccccc|cc|c}
\toprule
Method & \clean & \pegasus & \parrot & \dipper & \llmrewrite & \shortstack{Mean\\\tpro} & \shortstack{Mean\\\tprf} & \shortstack{Corp. PPL $\downarrow$\\(\pplratio)} \\
\midrule
\rowcolor{gray!15}
\multicolumn{9}{@{}l}{\textsc{Token-based watermarks}} \\
\addlinespace[2pt]
\KGW & \textbf{100.0} / \underline{99.6} / \textbf{100.0} & 97.8 / 76.0 / 91.6 & 98.6 / 84.0 / 93.8 & 93.0 / 47.0 / 70.6 & 70.0 / 8.2 / 23.6 & 63.0 & 75.9 & 8.7 (0.897) \\
\synthid & \textbf{100.0} / \textbf{100.0} / \textbf{100.0} & 91.2 / 55.6 / 73.0 & 90.6 / 58.6 / 72.6 & 73.5 / 18.2 / 31.8 & 57.2 / 6.8 / 16.6 & 47.8 & 58.8 & 7.7 (0.793) \\
\EXP & 99.8 / 99.2 / \underline{99.4} & 89.5 / 56.0 / 69.0 & 89.2 / 58.0 / 71.4 & 80.0 / 28.2 / 45.4 & 56.2 / 4.4 / 12.6 & 49.2 & 59.6 & \underline{4.5} (\underline{0.471}) \\
\SIR & \underline{99.9} / 97.0 / 99.2 & 95.1 / 76.4 / 88.6 & 96.2 / 75.2 / 87.2 & 93.0 / 61.6 / 79.0 & 82.9 / 17.0 / 42.6 & 65.4 & 79.3 & 7.6 (0.783) \\
\addlinespace[3pt]
\rowcolor{blue!10}
\multicolumn{9}{@{}l}{\textsc{Semantic watermarks}} \\
\addlinespace[2pt]
\semstamp & \underline{99.9} / 99.2 / \textbf{100.0} & 96.6 / 57.8 / 89.6 & 97.5 / 67.6 / 92.6 & 87.0 / 25.4 / 61.6 & 66.7 / 4.4 / 28.2 & 50.9 & 74.4 & 5.7 (0.591) \\
\ksemstamp & 99.4 / 96.4 / 99.2 & 94.7 / 45.2 / 74.4 & 96.6 / 52.4 / 82.8 & 85.9 / 16.8 / 46.6 & 71.7 / 4.4 / 16.6 & 43.0 & 63.9 & 6.1 (0.633) \\
\simmark & 99.2 / 89.8 / 96.2 & 98.2 / 72.8 / 92.2 & 98.7 / 79.0 / 93.0 & 95.3 / 43.0 / 77.4 & 84.2 / 12.4 / 37.6 & 59.4 & 79.3 & 6.1 (0.635) \\
\pmark & 99.4 / 91.0 / 97.4 & 95.6 / 34.8 / 84.8 & 98.0 / 59.4 / 90.6 & 87.1 / 20.6 / 57.6 & 74.6 / 3.4 / 24.6 & 41.8 & 71.0 & \textbf{3.9} (\textbf{0.401}) \\
\alimark & 99.8 / 98.0 / \underline{99.4} & 97.9 / 78.6 / 90.2 & 98.3 / 82.8 / 93.0 & 92.5 / 54.2 / 71.6 & 89.2 / 33.4 / 56.8 & 69.4 & 82.2 & 11.3 (1.168) \\
\samark & 95.2 / 60.2 / 87.0 & 95.3 / 62.6 / 84.8 & 96.4 / 68.2 / 88.8 & 93.3 / 49.8 / 75.4 & 92.1 / 43.6 / 74.6 & 56.9 & 82.1 & 6.8 (0.708) \\
\swordstamp & 99.7 / 93.2 / 99.2 & 99.1 / 86.8 / 97.0 & 99.3 / 87.4 / \underline{98.0} & 96.4 / 72.2 / 88.0 & 84.4 / 22.8 / 53.0 & 72.5 & 87.0 & 6.7 (0.691) \\
\kswordstamp & 99.7 / 87.0 / 95.4 & 99.0 / 69.2 / 90.4 & 99.2 / 75.8 / 92.4 & 97.7 / 59.0 / 80.6 & 95.5 / 39.8 / 70.0 & 66.2 & 85.8 & 6.8 (0.707) \\
\addlinespace[3pt]
\rowcolor{green!12}
\multicolumn{9}{@{}l}{\textsc{Ours}} \\
\addlinespace[2pt]
\semfield & 99.8 / 96.8 / 99.0 & \underline{99.4} / \underline{91.8} / \underline{98.0} & \underline{99.5} / \underline{92.6} / 97.8 & \underline{98.9} / \textbf{88.2} / \textbf{96.2} & \underline{96.2} / \underline{66.2} / \underline{84.6} & \underline{87.1} & \underline{95.1} & 6.8 (0.703) \\
\semfieldpl & \textbf{100.0} / \underline{99.6} / \textbf{100.0} & \textbf{99.7} / \textbf{95.4} / \textbf{99.0} & \textbf{99.8} / \textbf{96.2} / \textbf{99.4} & \textbf{99.1} / \underline{86.8} / \underline{95.8} & \textbf{97.8} / \textbf{68.6} / \textbf{89.6} & \textbf{89.3} & \textbf{96.8} & 6.9 (0.715) \\
\bottomrule
\end{tabular}%
} 
\end{table*}

\begin{table*}[h]
\centering
\scriptsize
\setlength{\tabcolsep}{2.0pt}
\renewcommand{\arraystretch}{1.10}
\caption{Fixed-threshold detection results for \qwen. The same cutoff is used for clean and all rewriting conditions. Best values in a column are bolded, and the second best are underlined.}
\label{tab:qwen3-1.7b-appendix-null10k-current}
\resizebox{\textwidth}{!}{%
\begin{tabular}{lccccc|cc|c}
\toprule
Method & \clean & \pegasus & \parrot & \dipper & \llmrewrite & \shortstack{Mean\\\tpro} & \shortstack{Mean\\\tprf} & \shortstack{Corp. PPL $\downarrow$\\(\pplratio)} \\
\midrule
\rowcolor{gray!15}
\multicolumn{9}{@{}l}{\textsc{Token-based watermarks}} \\
\addlinespace[2pt]
\KGW & \underline{99.9} / \underline{99.6} / \underline{99.8} & 96.8 / 79.2 / 89.6 & 98.8 / 88.2 / 96.2 & 91.4 / 44.4 / 67.8 & 76.2 / 9.8 / 27.4 & 64.2 & 76.2 & 10.9 (1.125) \\
\synthid & \textbf{100.0} / \textbf{99.8} / \underline{99.8} & 90.3 / 50.8 / 72.4 & 94.0 / 59.0 / 80.2 & 72.5 / 12.2 / 30.4 & 59.3 / 6.4 / 18.0 & 45.6 & 60.2 & 9.7 (1.007) \\
\EXP & \underline{99.9} / 99.4 / \underline{99.8} & 91.3 / 57.6 / 71.0 & 92.7 / 60.6 / 75.6 & 78.9 / 25.8 / 47.4 & 59.9 / 4.8 / 13.6 & 49.6 & 61.5 & \underline{5.2} (\underline{0.535}) \\
\SIR & \underline{99.9} / 98.4 / 99.6 & 95.6 / 83.6 / 90.0 & 98.4 / 88.4 / 94.0 & 95.2 / 61.2 / 79.4 & 81.7 / 18.8 / 38.6 & 70.1 & 80.3 & 10.4 (1.075) \\
\addlinespace[3pt]
\rowcolor{blue!10}
\multicolumn{9}{@{}l}{\textsc{Semantic watermarks}} \\
\addlinespace[2pt]
\semstamp & 99.6 / 94.4 / 99.0 & 94.5 / 50.0 / 83.8 & 96.5 / 64.4 / 89.8 & 80.8 / 16.4 / 48.0 & 64.5 / 4.0 / 24.0 & 45.8 & 68.9 & 8.5 (0.881) \\
\ksemstamp & 98.9 / 91.4 / 98.4 & 93.6 / 37.0 / 67.4 & 94.4 / 44.6 / 75.4 & 83.3 / 12.4 / 36.0 & 71.3 / 3.6 / 20.4 & 37.8 & 59.5 & 7.6 (0.792) \\
\simmark & 98.6 / 80.4 / 94.0 & 97.8 / 64.6 / 87.4 & 98.1 / 67.6 / 91.0 & 93.8 / 30.0 / 66.2 & 85.7 / 15.2 / 46.2 & 51.6 & 77.0 & 9.0 (0.930) \\
\pmark & 99.5 / 98.4 / 98.8 & 96.3 / 77.0 / 86.6 & 97.8 / 85.8 / 92.6 & 86.6 / 47.2 / 62.6 & 77.0 / 18.4 / 32.4 & 65.4 & 74.6 & \textbf{4.3} (\textbf{0.443}) \\
\alimark & \textbf{100.0} / 99.4 / \textbf{100.0} & 97.4 / 77.6 / 90.6 & 98.3 / 86.4 / 94.2 & 89.6 / 42.6 / 64.2 & 88.1 / 32.2 / 57.0 & 67.6 & 81.2 & 19.0 (1.970) \\
\samark & 96.9 / 67.2 / 92.4 & 96.8 / 68.9 / 89.4 & 97.6 / 73.9 / 90.0 & 95.6 / 56.0 / 82.5 & 93.9 / 47.0 / 76.9 & 62.6 & 86.2 & 8.5 (0.881) \\
\swordstamp & 99.8 / 96.4 / 99.4 & 99.1 / 88.4 / 97.4 & \underline{99.5} / 91.8 / \underline{98.4} & 96.9 / 69.1 / 89.8 & 85.2 / 21.5 / 50.0 & 73.4 & 87.0 & 9.3 (0.961) \\
\kswordstamp & 99.8 / 89.4 / 97.6 & 98.9 / 71.7 / 90.2 & 99.4 / 80.7 / 94.0 & 97.9 / 62.9 / 87.8 & \underline{96.9} / 43.8 / 75.3 & 69.7 & 89.0 & 9.0 (0.936) \\
\addlinespace[3pt]
\rowcolor{green!12}
\multicolumn{9}{@{}l}{\textsc{Ours}} \\
\addlinespace[2pt]
\semfield & 99.8 / 97.2 / 99.2 & \underline{99.4} / \underline{92.6} / \underline{97.6} & \underline{99.5} / \underline{93.6} / 97.8 & \underline{98.8} / \underline{86.4} / \underline{94.8} & 96.6 / \underline{69.0} / \underline{87.2} & \underline{87.8} & \underline{95.3} & 8.6 (0.889) \\
\semfieldpl & \underline{99.9} / \underline{99.6} / \underline{99.8} & \textbf{99.7} / \textbf{95.2} / \textbf{99.0} & \textbf{99.8} / \textbf{96.2} / \textbf{99.4} & \textbf{98.9} / \textbf{86.8} / \textbf{95.4} & \textbf{97.9} / \textbf{78.0} / \textbf{89.8} & \textbf{91.2} & \textbf{96.7} & 8.8 (0.910) \\
\bottomrule
\end{tabular}%
} 
\end{table*}

\begin{table*}[h]
\centering
\scriptsize
\setlength{\tabcolsep}{2.0pt}
\renewcommand{\arraystretch}{1.10}
\caption{Fixed-threshold detection results for \opt. The same cutoff is used for clean and all rewriting conditions. Best values in a column are bolded, and the second best are underlined.}
\label{tab:opt-1.3b-appendix-null10k-current}
\resizebox{\textwidth}{!}{%
\begin{tabular}{lccccc|cc|c}
\toprule
Method & \clean & \pegasus & \parrot & \dipper & \llmrewrite & \shortstack{Mean\\\tpro} & \shortstack{Mean\\\tprf} & \shortstack{Corp. PPL $\downarrow$\\(\pplratio)} \\
\midrule
\rowcolor{gray!15}
\multicolumn{9}{@{}l}{\textsc{Token-based watermarks}} \\
\addlinespace[2pt]
\KGW & 99.8 / \underline{99.8} / \underline{99.8} & 97.6 / 81.4 / 91.6 & 98.1 / 87.2 / 95.4 & 91.2 / 47.2 / 65.8 & 68.1 / 10.4 / 22.0 & 65.2 & 74.9 & 8.5 (0.876) \\
\synthid & \textbf{100.0} / \textbf{100.0} / \textbf{100.0} & 90.0 / 53.4 / 68.6 & 90.4 / 58.0 / 72.2 & 76.1 / 21.4 / 36.8 & 58.3 / 7.6 / 16.2 & 48.1 & 58.8 & 8.0 (0.832) \\
\EXP & 99.7 / 99.2 / 99.6 & 88.6 / 55.5 / 70.7 & 89.0 / 58.5 / 73.1 & 76.2 / 27.9 / 42.7 & 55.5 / 5.2 / 14.2 & 49.3 & 60.1 & \underline{5.2} (\underline{0.541}) \\
\SIR & 99.5 / 94.6 / 98.0 & 96.4 / 68.8 / 82.0 & 97.4 / 70.0 / 83.6 & 91.5 / 31.8 / 52.4 & 85.2 / 5.8 / 18.0 & 54.2 & 66.8 & 8.3 (0.856) \\
\addlinespace[3pt]
\rowcolor{blue!10}
\multicolumn{9}{@{}l}{\textsc{Semantic watermarks}} \\
\addlinespace[2pt]
\semstamp & \underline{99.9} / 98.2 / \textbf{100.0} & 95.1 / 53.0 / 86.4 & 96.8 / 60.2 / 89.6 & 86.1 / 20.4 / 59.0 & 66.4 / 5.6 / 24.0 & 47.5 & 71.8 & 8.5 (0.884) \\
\ksemstamp & 99.7 / 92.6 / \underline{99.8} & 94.7 / 42.0 / 75.2 & 95.4 / 41.2 / 76.4 & 87.0 / 14.2 / 46.2 & 69.6 / 3.4 / 14.6 & 38.7 & 62.4 & 6.5 (0.668) \\
\simmark & 98.6 / 84.6 / 94.0 & 97.1 / 67.4 / 88.2 & 97.5 / 70.6 / 89.2 & 93.8 / 41.6 / 72.4 & 86.3 / 18.6 / 47.2 & 56.6 & 78.2 & 6.7 (0.698) \\
\pmark & \underline{99.9} / 99.4 / \underline{99.8} & 95.7 / 77.0 / 86.4 & 99.1 / 93.2 / 97.0 & 86.2 / 49.0 / 63.4 & 69.3 / 13.8 / 27.4 & 66.5 & 74.8 & \textbf{4.3} (\textbf{0.442}) \\
\alimark & 99.8 / 99.4 / 99.4 & 98.5 / 83.4 / 94.6 & 98.8 / 90.6 / 96.0 & 92.6 / 51.4 / 71.6 & 91.1 / 42.2 / 64.8 & 73.4 & 85.3 & 8.5 (0.878) \\
\samark & 98.9 / 70.6 / 95.6 & 98.7 / 74.0 / 94.0 & 99.1 / 76.4 / 96.0 & 96.8 / 54.4 / 83.6 & 95.4 / 46.6 / 78.8 & 64.4 & 89.6 & 6.7 (0.697) \\
\swordstamp & \textbf{100.0} / \underline{99.8} / \underline{99.8} & \underline{99.5} / 92.0 / 98.2 & \underline{99.7} / \underline{95.0} / \underline{99.2} & 97.5 / 76.6 / 91.8 & 86.5 / 20.8 / 49.2 & 76.8 & 87.6 & 6.7 (0.689) \\
\kswordstamp & \underline{99.9} / 90.6 / 98.4 & 99.1 / 72.4 / 92.0 & 99.5 / 79.6 / 93.8 & 98.2 / 63.0 / 84.8 & 96.1 / 33.8 / 67.6 & 67.9 & 87.3 & 7.0 (0.725) \\
\addlinespace[3pt]
\rowcolor{green!12}
\multicolumn{9}{@{}l}{\textsc{Ours}} \\
\addlinespace[2pt]
\semfield & 99.8 / 97.0 / 99.6 & \underline{99.5} / \underline{92.2} / \underline{98.6} & 99.5 / 91.8 / 98.8 & \underline{99.0} / \textbf{87.2} / \textbf{96.4} & \underline{96.6} / \underline{67.4} / \underline{86.2} & \underline{87.1} & \underline{95.9} & 6.3 (0.656) \\
\semfieldpl & \textbf{100.0} / \textbf{100.0} / \textbf{100.0} & \textbf{99.8} / \textbf{95.0} / \textbf{99.6} & \textbf{99.8} / \textbf{95.8} / \textbf{99.6} & \textbf{99.1} / \underline{85.0} / \underline{95.8} & \textbf{97.5} / \textbf{71.8} / \textbf{88.0} & \textbf{89.5} & \textbf{96.6} & 6.4 (0.661) \\
\bottomrule
\end{tabular}%
} 
\end{table*}
\section{\semfield Algorithms}
\label{app:semfield-algo}

\begin{figure*}[h]
\centering
\scriptsize

\ruledalgorithm{0.47\textwidth}{Algorithm 1: Generation}{
\begin{algorithmic}[1]
\Require Prompt $p$, model $\mathcal{M}$, key $K$,
candidate count $Q$, sentence budget $T$,
variant $v\in\{\mathrm{SF},\mathrm{PL}\}$
\Ensure Watermarked document $D$

\State $\omega_K\gets\operatorname{KeyedGaussian}(K)$
\State $A_0\gets\mathbf{0}$
\State $D\gets()$

\For{$t=1,\ldots,T$}
    \State Sample completed candidates
    $C_t=\{c_{t,1},\ldots,c_{t,Q}\}$
    conditioned on $p$ and $D$
    \For{$q=1,\ldots,Q$}
        \State $x_{t,q}\gets
        \operatorname{Normalize}(\operatorname{Embed}(c_{t,q}))$
        \State $A_{t,q}\gets A_{t-1}+x_{t,q}$
        \State $G_{t,q}\gets
        \dfrac{\omega_K^{\mathsf T}A_{t,q}}
        {\lVert A_{t,q}\rVert_2}$
    \EndFor

    \If{$v=\mathrm{SF}$}
        \State $q_t^\star\gets
        \operatorname*{arg\,max}_{1\leq q\leq Q}G_{t,q}$
    \ElsIf{$t=1$}
        \State $q_t^\star\gets
        \operatorname*{arg\,max}_{1\leq q\leq Q}|G_{t,q}|$
        \State $\sigma\gets
        \operatorname{sign}(G_{t,q_t^\star})$
        \State Set $\sigma\gets+1$ if $G_{t,q_t^\star}=0$
    \Else
        \State $q_t^\star\gets
        \operatorname*{arg\,max}_{1\leq q\leq Q}
        \sigma G_{t,q}$
    \EndIf

    \State $c_t\gets c_{t,q_t^\star}$
    \State $A_t\gets A_{t,q_t^\star}$
    \State Append $c_t$ to $D$
\EndFor

\State \Return $D$
\end{algorithmic}
}
\hfill
\ruledalgorithm{0.47\textwidth}{Algorithm 2: Detection}{
\begin{algorithmic}[1]
\Require Document $D=(c_1,\ldots,c_T)$, key $K$,
variant $v\in\{\mathrm{SF},\mathrm{PL}\}$,
variant-specific threshold $\tau_v$
\Ensure Detection decision

\State $\omega_K\gets\operatorname{KeyedGaussian}(K)$
\State $A\gets\mathbf{0}$

\For{$t=1,\ldots,T$}
    \State $x_t\gets
    \operatorname{Normalize}(\operatorname{Embed}(c_t))$
    \State $A\gets A+x_t$
\EndFor

\State $Z_K\gets
\dfrac{\omega_K^{\mathsf T}A}{\lVert A\rVert_2}$

\If{$v=\mathrm{SF}$}
    \State $H_K\gets Z_K$
\Else
    \State $H_K\gets |Z_K|$
\EndIf

\If{$H_K\geq\tau_v$}
    \State \Return watermarked
\Else
    \State \Return non-watermarked
\EndIf
\end{algorithmic}
}

\caption{Generation and detection for regular \semfield
($v=\mathrm{SF}$) and polarity-locked \semfieldpl
($v=\mathrm{PL}$).}
\label{fig:semfield-algorithms}
\end{figure*}

\end{document}